\documentclass[acmsmall]{acmart}
\renewcommand\footnotetextcopyrightpermission[1]{} 
\usepackage{booktabs} 
\usepackage{amsmath,amsthm, amsopn}
\usepackage{verbatim}
\usepackage{tikz}
\usepackage{algorithm}
	\usetikzlibrary{shapes.geometric, calc, positioning, fit, shapes.misc,positioning}
\usepackage{bm}
\usepackage{url}
\usetikzlibrary{calc}
\usepackage[noend]{algpseudocode}
\usepackage{todonotes}
\usepackage[appendix=append]{apxproof}
\usepackage{enumitem}

\newcommand{\DelayClin}{{\sf DelayC_{lin}}}

\newcommand{\CDlin}{{\sf CD\!\circ\! Lin}}
\newcommand{\linpartial}{{\sf Lin_{Partial}}}
\newcommand{\encode}[1]{||#1||}
\newcommand{\qDecide}[1]{\text{\sc Decide}\langle#1\rangle}
\newcommand{\pEnum}[1]{\text{\sc Enum}\langle#1\rangle}

\newcommand{\idom}{{\it dom}}
\newcommand{\arity}{\operatorname{arity}}

\newcommand{\var}{{\it var}}
\newcommand{\ivar}{{\it var}}
\newcommand{\free}{\operatorname{free}}
\newcommand{\atoms}{\operatorname{atoms}}
\newcommand{\full}{\operatorname{full}}

\newcommand{\sigdiff}{\sigma_\neq}
\newcommand{\taudiff}{\tau_\neq}

\newcommand{\matmul}{\textsc{mat-mul}}
\newcommand{\hyperclique}{\textsc{hyperclique}}
\newcommand{\fourclique}{4\textsc{-clique}}
\newcommand{\CDY}{CDY}

\newcommand{\calA}{{\mathcal A}}

\newcommand{\calH}{{\mathcal H}}
\newcommand{\calO}{{\mathcal O}}
\newcommand{\calR}{{\mathcal R}}
\newcommand{\calS}{{\mathcal S}}
\newcommand{\calU}{{\mathcal U}}

\theoremstyle{plain}
\newtheoremrep{theorem}{Theorem}
\newtheoremrep{lemma}[theorem]{Lemma}
\newtheoremrep{proposition}[theorem]{Proposition}
\newtheoremrep{corollary}[theorem]{Corollary}

\newtheorem*{claim*}{Claim}
\newtheorem{claim}{Claim}
\newtheorem{remark}{Remark}
\theoremstyle{definition}
\theoremstyle{plain}

\begin{document}

\title{On the Enumeration Complexity of Unions of Conjunctive Queries}

\author{Nofar Carmeli}
\orcid{0003-0673-5510}
\affiliation{
	\institution{Technion - Israel Institute of Technology}
	\city{Haifa}
	\state{Israel}
}
\email{snofca@cs.technion.ac.il}

\author{Markus Kr\"oll}
\affiliation{
	\institution{TU Wien}
	\city{Vienna}
	\state{Austria}
}
\email{kroell@dbai.tuwien.ac.at}

\begin{abstract}
We study the enumeration complexity of Unions of Conjunctive Queries (UCQs).
We aim to identify the UCQs that are tractable in the sense that the answer tuples
can be enumerated with a linear preprocessing phase and a constant delay between every successive tuples. It has been established that, in the absence of self-joins and under conventional complexity assumptions, the CQs that admit such an evaluation are precisely the free-connex ones.
A union of tractable CQs is always tractable. We generalize  the notion of free-connexity from CQs to UCQs,
thus showing that some unions containing intractable CQs are, in fact, tractable. Interestingly, some unions consisting of only intractable CQs are tractable too.
We show how to use the techniques presented in this article also in settings where the database contains cardinality dependencies (including functional dependencies and key constraints) or when the UCQs contain disequalities. 
The question of finding a full characterization of the tractability of UCQs remains open.
Nevertheless, we 
prove that for several classes of queries, free-connexity fully captures the tractable UCQs.
\end{abstract}

\maketitle
\thispagestyle{empty}
\pagestyle{empty}


\section{Introduction}

Evaluating a query $Q$ over a database instance $D$ is a fundamental and well-studied
problem in database management systems. Most of the complexity results dealing
with query evaluation so far aim to solve a decision or counting variant of it.
Starting with Durand and Grandjean in 2007~\cite{Durand:constantdelay},
there has been a renewed interest in examining the enumeration problem of all answers to a query, focusing on fine-grained complexity~\cite{DBLP:conf/pods/SchweikardtSV18, DBLP:conf/pods/NiewerthS18, DBLP:conf/pods/FlorenzanoRUVV18, DBLP:conf/icdt/CarmeliK18}.
When evaluating a non-Boolean query over a database, the number of results may be larger than the size of the database itself. Enumeration complexity offers specific measures for the hardness of such problems.
In terms of data complexity, the best time guarantee we can hope for is to output all answers with
a constant delay between consecutive answers. 
In the case of query evaluation, this enumeration phase comes after a linear preprocessing phase required to read the database and decide the existence of a first answer.
The enumeration class achieving these time bounds is denoted by $\DelayClin$.
Hereafter, we refer to queries in $\DelayClin$ as tractable, and queries outside of this class as intractable.

Results by Bagan et al.~\cite{bdg:dichotomy} and Brault-Baron~\cite{bb:thesis} form a dichotomy that fully classifies which self-join-free Conjunctive Queries (CQs)
are in the class $\DelayClin$ based on the structure of the query: the class of \emph{free-connex} queries is exactly the class that admits tractable enumeration.
In the years following
this dichotomy, much work has been conducted to achieve similar results for other classes of
queries~\cite{Segoufin:survey}. 
Unions of CQs (UCQs) are a natural extension of CQs, as they describe the union of the answers to several CQs.
UCQs form an important class of queries, as it captures the positive fragment of relational algebra.
Previous work which implies results on the enumeration complexity of UCQs imposes strong restrictions on
the underlying database~\cite{segoufin_et_al:LIPIcs:2017:7060, DBLP:conf/pods/SchweikardtSV18}.
We aim to understand the enumeration complexity of UCQs without such
restrictions and based solely on their structure.

Using known methods~\cite{strozecki:thesis}, it can be shown that a union of tractable problems is again tractable.
However, what happens if some CQs of a union are tractable while others are not?
Intuitively, one might be tempted to expect a union of enumeration problems to be harder
than a single problem
within the union, making such a UCQ intractable as well.
As we will show, this is not necessarily the case.
\begin{example}
Let $Q=Q_1\cup Q_2$ with 
\begin{align*}
Q_1(x,y)&\leftarrow R_1(x,y),R_2(y,z),R_3(z,x)\text{ and }\\ 	
Q_2(x,y)&\leftarrow R_1(x,y),R_2(y,z).
\end{align*}
Even though $Q_1$ is hard while $Q_2$ is easy, a closer look shows that $Q_2$ contains $Q_1$. This means that $Q_1$ is redundant, and the entire union is equivalent to the easy $Q_2$.
\end{example}
To avoid cases like these, where the UCQ can be translated to a simpler one, it makes sense to consider non-redundant unions.
It was claimed that in all cases of a non-redundant union containing an intractable CQ, the UCQ is intractable too~\cite{DBLP:conf/icdt/BerkholzKS18}. The following is a counter example which refutes this claim.

\begin{example}\label{example:first}
Let $Q=Q_1\cup Q_2$ with 
\begin{align*}
Q_1(x,y,w)&\leftarrow R_1(x,z),R_2(z,y),R_3(y,w)\text{ and }\\ 	
Q_2(x,y,w)&\leftarrow R_1(x,y),R_2(y,w).
\end{align*}
According to the dichotomy of Bagan et al.~\cite{bdg:dichotomy}, the enumeration problem
for $Q_2$ is in $\DelayClin$, while 
$Q_1$ is intractable.
Yet, it turns out that $Q$ is in fact in $\DelayClin$. 
The reason is that, since $Q_1$ and $Q_2$ are evaluated over the same database $I$, we can use $Q_2(I)$ to find $Q_1(I)$.
We can compute $Q_2(I)$ efficiently, and try to extend every such solution to solutions of $Q_1$ with a constant delay:
for every new combination $(a,c)$ of an output $(a,b,c)\in Q_2(I)$,
we find all $d$ values with $(b,d)\in R_3^I$ and then output the solution $(a,c,d)\in Q_1(I)$.
The new combinations $(a,c)$ can be detected by storing a lookup table containing $(a,c)$ for every answer $(a,b,c)\in Q_2(I)$ generated. This table can be accessed in constant time in the RAM model.
Avoiding duplicates in case $Q_1$ and $Q_2$ share answers requires also storing the previously seen answers, and it is discussed in more details in Section~\ref{sec:positive}.
Intuitively, the source of intractability for $Q_1$ is the join of $R_1$ with $R_2$ as we need to avoid duplicates that originate in different $z$ values. The union is tractable since $Q_2$ returns exactly this join.
\qed
\end{example}

As the example illustrates, to compute the answers to a UCQ in an efficient way, 
it is not enough to view it as a union of isolated
instances of CQ enumeration. In fact, this task requires an
understanding of the interaction between several queries.
Example~\ref{example:first} shows that the presence of an easy query within
the union may give us enough time to compute auxiliary data structures, which we can then add to the
hard queries in order to enumerate their answers as well.
In Example~\ref{example:first}, we can assume we have a ternary relation holding the result of $Q_2$. Then, adding the auxiliary atom $R_{Q_2}(x,z,y)$ to $Q_1$ results in a tractable structure.
We generalize this observation and introduce the concept of \emph{union-extended} queries.
We then use union extensions as a central tool for evaluating the enumeration complexity of UCQs, as the structure
of such queries has implications on the tractability of the UCQ.

Interestingly, this approach can be taken a step further: We show that the concept of
extending the union by auxiliary atoms can even be used to efficiently enumerate the answers
of UCQs that \emph{only} contain hard queries. 
By lifting the concept of free-connex queries from CQs to UCQs via union-extensions,
we show that free-connex UCQs are always tractable.
This gives us a sufficient global condition for membership in $\DelayClin$ beyond any classification of individual CQs.

Finding a full characterization of the tractability of UCQs with respect to $\DelayClin$ remains an open problem.
Nevertheless, we 
prove that for several classes of queries, free-connexity fully captures the tractable UCQs. 
A non-free-connex union of two CQs is intractable in the following cases:
both CQs are intractable, or they both represent the same CQ up to a different projection.
The hardness results presented here use problems with well-established assumptions on the lower bounds, such as Boolean matrix-multiplication~\citep{MatrixMultiplication} or finding a clique or a hyperclique in a graph~\cite{DBLP:conf/soda/LincolnWW18}.

Why is establishing lower bounds on UCQ evaluation, even when it contains only two CQs, a fundamentally more challenging problem than the one for CQs?
In the case of CQs, 
hardness results are often shown by
reducing a computationally hard problem to the task of answering a query. The reduction encodes the hard problem to the relations of a self-join-free CQ, such
that the answers of the CQ correspond to an answer of this
problem~\cite{bdg:dichotomy, bb:thesis, DBLP:conf/pods/BerkholzKS17, DBLP:conf/icdt/CarmeliK18}.
However, using such an encoding for CQs within a union does not always work. Similarly to the case of CQs with
self-joins, relational symbols that appear multiple times within a query can interfere with the reduction.
Indeed, when encoding a hard problem to an intractable CQ within a union, a different CQ in the union evaluates over the same relations, and may also produce answers.
A large number of such supplementary answers, with constant delay per answer, accumulates to a long delay until we obtain the answers that correspond to the computationally hard problem.
If this delay is larger than the lower bound we assume for the hard problem, we cannot conclude that the UCQ is intractable.

The lower bounds presented in this article are obtained either by identifying classes of UCQs for which we can use similar reductions to the ones used for CQs,
or by introducing alternative reductions.
As some cases remain unclassified, we devote a section to inspecting such UCQs, and describing the challenges that will need to be resolved in order to achieve a full classification.

Most of the results in this work refer to the classification of the hardness of UCQs over general database instances with respect to the time it takes to answer them with no space restrictions.
Nonetheless, we also discuss the implications of these results when the settings are slightly modified.
The first modification we consider is when the database instances are not general, as they are restricted by the schema.
\emph{Cardinality dependencies} are commonly given as part of the database schema to capture the fact that, since we know what our database represents, we know of restrictions on the combinations of values that the relations may hold.
Functional Dependencies and key constraints are a special case of cardinality dependencies.
In the presence of cardinality dependencies, additional CQs are tractable~\cite{DBLP:conf/icdt/CarmeliK18}.
We show that by combining the approach introduced in this article with the one for identifying tractable CQs in the presence of cardinality dependencies, we can also find additional tractable UCQs over such schemas.

We also inspect UCQs with disequalities. In the case of CQs with disequalities, the disequalities have no effect on the enumeration complexity~\cite{bdg:dichotomy}. That is, one can simply ignore the disequalities, and the remaining CQ is free-connex if and only if the query is tractable (under the same assumptions used for the hardness of CQs). This irrelevance of disequalities to the enumeration complexity holds also in the case that there are cardinality dependencies in the schema~\cite{DBLP:conf/icdt/CarmeliK18}.
A natural question is: does this happen also with UCQs? We show that the answer is no: a tractable UCQ may become intractable when adding disequalities. We then show how to identify easy UCQs with disequalities using the same techniques we introduce for UCQs without them.

The final modification we consider is that of restricted space.
We introduce the time complexity measure of \emph{linear partial time}, which is more relaxed compared to linear preprocessing and constant delay, and discuss the connection between the two measures. We argue that from a theoretical point of view, when given access to enough space, the two time  complexity measures are equivalent, and that from a practical point of view, an algorithm with only the relaxed requirements is sometimes preferable.
We discuss the space consumption used in our approach, and demonstrate that in some cases it can be reduced to constant additional space.
The question of when exactly this can be done remains open, but we claim in favor of using the measure of linear partial time as an intermediate step for future research.

This article is organized as follows: In Section~\ref{sec:Preliminaries} we provide definitions and state results that we use. Section~\ref{sec:positive} formalizes how CQs within a union can make each other easier, defines free-connex UCQs, and proves that free-connex UCQs are in $\DelayClin$. In Section~\ref{sec:lower-bounds} we prove conditional lower bounds and conclude a dichotomy for some classes of UCQs. Section~\ref{sec:dichotomy} discusses the future steps required for a full classification, and demonstrates 
examples of queries of unknown complexity.
Section~\ref{sec:extensions} inspects the implications of this work with modified settings: in the presence of cardinality dependencies, with disequalities and with restricted space.
Concluding remarks are given in Section~\ref{sec:conclusions}.

A short version of this article appeared in the 38th ACM SIGMOD-SIGACT-SIGAI Symposium on Principles of Database Systems (PODS19)~\cite{DBLP:conf/pods/CarmeliK19}. This article has the following additions. First, it closes some gaps by
presenting full proofs to the results stated in the conference version. Second, it includes
Section~\ref{sec:extensions}, which discusses the implications of this work in additional settings.

\section{Preliminaries} \label{sec:Preliminaries}

In this section we provide preliminary definitions as well as state results that we will use throughout this article.

\subsubsection*{Unions of Conjunctive Queries} 
A {\em schema} $\calS$ is a set of {\em relational symbols} $\{R_1,\ldots,R_n\}$.
We denote the {\em arity} of a relational symbol $R_i$
as $\arity(R_i)$. 
Let $\idom$ be a finite set of constants. A database $I$ over schema $\calS$ is called an 
{\em instance} of $\calS$, and it consists of a finite relation $R_i^I\subseteq\idom^{\arity(R_i)}$ for every relational symbol $R_i\in\calR$.

Let $\var$ be a set of variables disjoint from $\idom$.
A {\em Conjunctive Query} (CQ) over schema $\calS$ is an expression of the form 
$Q(\vec{p}) \leftarrow R_1(\vec{v}_1), \dots, R_m(\vec{v}_m)$,
where $R_1,\ldots,R_m$ are relational symbols of $\calS$, the tuples $\vec{p}, \vec{v}_1,\ldots, \vec{v}_m$ hold variables, and every variable in $\vec{p}$ appears in at least one of $\vec{v}_1,\ldots, \vec{v}_m$. 
We often denote this query as $Q$. Define the variables of $Q$ as
$\var(Q)=\bigcup_{i=1}^{m} \vec{v}_i$,
and define the {\em free variables}
of $Q$ as $\free(Q)=\vec{p}$.
The variables of $Q$ that are not free are called \emph{existential}.
We call $Q(\vec{p})$ the head of $Q$, and the atomic formulas $R_i(\vec{v}_i)$ are called {\em atoms}.
We further use $\atoms(Q)$ to denote the set of atoms of Q.
We denote by $\full(Q)$ the full version of a CQ $Q$ obtained by adding all variables to the head of $Q$.
A CQ is said to be {\em self-join-free} if no relational symbol appears in more than one atom.

The {\em evaluation} $Q(I)$ of a CQ $Q$ over a database $I$ is 
the set of all \emph{mappings} $\mu|_{\free(Q)}$ such that $\mu$ is a homomorphism from 
$R_1(\vec{v}_1), \dots , R_m(\vec{v}_m)$ into $I$, and $\mu|_{\free(Q)}$ is the restriction (or projection) of $\mu$ to the variables $\free(Q)$.
We refer to such mappings as \emph{answers} to $Q$ over $I$. To ease notation, we sometimes specify an answer as a tuple $\vec{u}$, and mean that the answer is the mapping that maps $\free(Q)$ to $\vec{v}$.
A CQ $Q_1$ is \emph{contained} in a CQ $Q_2$, denoted $Q_1\subseteq Q_2$, if for every instance $I$, $Q_1(I)\subseteq Q_2(I)$.
A \emph{homomorphism} from $Q_2$ to $Q_1$ is a mapping $h:\ivar(Q_2)\rightarrow\ivar(Q_1)$ such that:
(1) for every atom $R(\vec{v})$ of $Q_2$, $R(h(\vec{v}))$ is an atom in $Q_1$; (2) $h(\free(Q_2))=\free(Q_1)$.

A \emph{Union of Conjunctive Queries (UCQ)} $Q$ is a finite set of CQs, denoted $Q=\bigcup_{i=1}^{\ell} Q_i$, where $\free(Q_{i_1})=\free(Q_{i_2})$ for all $1\leq {i_1},{i_2}\leq \ell$. Semantically, $Q(I)=\bigcup_{i=1}^{\ell} Q_i(I)$.
We say that $Q$ is non-redundant if there are no two different CQs $Q_1$ and $Q_2$ in $Q$ such that there is a homomorphism from $Q_2$ to $Q_1$.
We assume that UCQs are non-redundant; otherwise, an equivalent non-redundant UCQ can be obtained by removing the redundant CQs.
Given a UCQ $Q$ and a database instance $I$, we denote by $\qDecide{Q}$ 
the problem of deciding whether $Q(I)\neq\emptyset$.

\subsubsection*{Hypergraphs}

A {\em hypergraph} $\calH=(V,E)$ is a set $V$ of {\em vertices} and a set $E$ of non-empty subsets of $V$ called {\em hyperedges} (sometimes {\em edges}).
A \emph{join tree} of a hypergraph $\calH=(V,E)$ is a tree $T$ where the nodes are the hyperedges of $\calH$, and the {\em running intersection} property holds, namely: for all $u \in V$ the set $\{e \in E \mid u \in e\}$ forms a connected subtree in $T$. 
A hypergraph $\calH$ is {\em acyclic} if there exists a join tree for $\calH$.

Two vertices in a hypergraph are {\em neighbors} if they appear in a common edge.
A {\em cycle} in a hypergraph is a sequence of vertices that starts and ends in the same vertex and every two consecutive vertices are neighbors.
A hypergraph {\em is chordal} if every cycle with at least four distinct vertices has an edge connecting two nonconsecutive vertices of the cycle.
Every acyclic graph is chordal~\cite{beeri1983desirability}.

A {\em clique} of a hypergraph is a set of vertices, which are pairwise neighbors in $\calH$.
If every edge in $\calH$ has exactly $k$ vertices, we call $\calH$ $k$-{\emph uniform}.
An $l$-{\em hyperclique} in a $k$-uniform hypergraph $\calH$ is a set $V'$ of $l>k$ vertices,
such that every subset of $V'$ of size $k$ forms a hyperedge.

A hypergraph $\calH'$ is an \emph{inclusive extension} of $\calH$
if every edge of $\calH$ appears in $\calH'$, and every edge of $\calH'$ is a subset of some edge in $\calH$.
A tree $T$ is an \emph{ext-$S$-connex tree} for a hypergraph $\calH$ if:
(1) $T$ is a join-tree of an inclusive extension of $\calH$, and
(2) there is a subtree $T'$ of $T$ that contains exactly the variables $S$
~\cite{bdg:dichotomy} (see Figure~\ref{figure:ext-connex}).

\begin{figure}
\begin{tikzpicture}
[
    he/.style={draw, rounded corners,inner sep=0pt},        
    ce/.style={draw,dashed, rounded corners=2pt}, 
]

\node (z) at (0,0) {$z$};
\node (w) at (1,0) {$w$};
\node (x) at (0,1) {$x$};
\node (y) at (1,1) {$y$};
\node (v) at (2,0) {$v$};

\node[he, fit = (x) (y)] {};
\node [he, fit = (w) (v)] {};

\node [fit = (z) (y) (w)] (zyw) {};										
\draw [rounded corners = 3pt] 
								($(z.south west) + (-0.05,-0.05)$) -- ($(z.north west) + (-0.05,-0.05)$) --
								 ($(y.north west) + (0.05,0.05)$) -- ($(y.north east) + (0.05,0.05)$) --
								  ($(w.south east) + (0.05,-0.05)$) -- cycle;
													 

\node (wyz) at (4,0) {$w, y, z$};
\node[he, fit = (wyz)] {};

\node (yz) at (4, 0.7) {$y, z$};
\node[he, fit = (yz)] {};					

\node (xy) at (4, 1.4) {$x,y$};
\node[he, fit = (xy)] {};			

\node (vw) at (5.3, 0) {$v, w$};
\node[he, fit = (vw)] {};			
								 
\node[ce, fit = (xy) (yz)] {};								 
								 
\draw (xy) -- (yz);								 
\draw (yz) -- (wyz);										 
\draw (wyz) -- (vw);		


\node(H) at (-1, 0.7) {$\calH:$};
\node (T) at (3,0.7) {$T:$};
\end{tikzpicture}
\caption{$T$ is an ext-$\{x,y,z\}$-connex tree for hypergraph $\calH$.}
\label{figure:ext-connex}
\end{figure}

\subsubsection*{Classes of CQs}
We associate a hypergraph $\calH(Q)=(V,E)$ to a CQ $Q$ where the vertices are the variables of $Q$, and every hyperedge is a set of variables occurring in a single atom of $Q$. That is, $E=\{\{v_1,\ldots,v_n\} \mid R_i(v_1,\ldots,v_n)\in\atoms(Q)\}$. With a slight abuse of notation, we identify atoms of $Q$ with edges of $\calH(Q)$.
A CQ $Q$ is said to be {\em acyclic} if $\calH(Q)$ is acyclic

A CQ $Q$ is {\em $S$-connex} if $\calH(Q)$ has an ext-$S$-connex tree, and it is \emph{free-connex} if it has an ext-$\free(Q)$-connex tree~\cite{bdg:dichotomy}. Equivalently, $Q$ is free-connex if both $Q$ and $(V, E\cup\{\free(Q)\})$ are acyclic~\cite{bb:thesis}.
A {\em free-path} in a CQ $Q$ is a sequence of variables $(x,z_1,\ldots,z_k,y)$ with $k\geq 1$, such that:
(1) $\{x,y\}\subseteq \free(Q)$
(2) $\{z_1,\ldots,z_k\}\subseteq V\setminus\free(Q)$
(3) It is a {\em chordless path} in $\calH(Q)$: that is, every two succeeding variables are neighbors in $\calH(Q)$, but no two non-succeeding variables are neighbors.
An acyclic CQ has a free-path iff it is not free-connex~\cite{bdg:dichotomy}.

\subsubsection*{Computational Model}

In this article, we adopt the \emph{Random Access Machine} (RAM) model with uniform-cost measure and word length $\Theta(\log(n))$ on input of size $n$. Operations such as addition of the values of two registers or concatenation can be performed in constant time. In contrast to the Turing model of computation, the RAM model with uniform-cost measure can retrieve the content of any register via its unique address in constant time.

The input to our problems is a database instance $I$. We denote by $\encode{o}$ the size of an object $o$ (i.e., the number of integers required to store it), whereas $|o|$ is its cardinality. Let $I$ be a database over a schema $\calS=(\calR,\Delta)$.
We assume the input database is given by the reasonable encoding suggested by Flum et al.~\cite{DBLP:journals/jacm/FlumFG02}. Thus, the input is of size $\encode{I}= 1 + |\idom| + |\calR| + \sum_{R\in\calR}{\arity(R)|R^{I}|}$ over integers bounded by $\max\{|\idom|,max_{R\in\calR}{|R^I|}\}$.

As we assume that the word length is $\Theta(\log(\encode{I}))$, the values stored in registers are at most $\encode{I}^c$ for some fixed integer $c$. As a consequence, and since we assume that the values may correspond to addresses, the amount of available memory is polynomial in $\encode{I}$.
Note that even in more restrictive variants of the RAM model, the size of the registers and the available memory depend on the size of the input~\cite{DBLP:journals/amai/Grandjean96}. In particular, since we assume that the input fits in memory and each address fits in a register, the registers must contain at least $\log(n)$ bits for an input that requires $n$ registers.

The RAM model enables the construction of large lookup tables that can be queried within constant time. If the available memory is as large as the number of different possible keys, the lookup table can be as simple as initializing the memory to False, and setting the addresses matching the keys to be True.
Note that the time for initialization is not an issue due to constant time initialization techniques~\cite{MoretShapiro}.
Solutions such as binary search trees and hash tables can reduce the memory consumption significantly.
For example, a lookup table based on a binary search tree only requires linear memory in its content at the cost of a logarithmic factor to the time.
In this work, we wish to identify what can and cannot be done within certain time bounds without restrictions on the memory consumption. Thus, we assume in our analysis that the available memory is very large and that we can access lookup tables of polynomial size in constant time. 

\subsubsection*{Enumeration Complexity}
Given a finite alphabet $\Sigma$ and binary relation $R\subseteq\Sigma^*\times\Sigma^*$, the {\em enumeration problem} $\pEnum{R}$ is: given an instance $x\in\Sigma^*$, output all $y\in\Sigma^*$ such that $(x,y)\in R$.
Such $y$ values are often called \emph{solutions} or \emph{answers} to $\pEnum{R}$.
An \emph{enumeration algorithm} $\calA$ for $\pEnum{R}$ is a RAM that solves $\pEnum{R}$ without repetitions.
We say that $\calA$ enumerates $\pEnum{R}$ with {\em delay} $d(|x|)$ if the time before the first output, the time between any two consecutive outputs, and the time between the last output and termination are each bounded by $d(|x|)$.
Sometimes we wish to relax the requirements of the delay before the first answer, and specify a {\em preprocessing} time $p(|x|)$. In this case, the time before the first output is only required to be bounded by $p(|x|)$.
The enumeration class $\DelayClin$ is defined as the class of all enumeration problems $\pEnum{R}$
which have an enumeration algorithm $\calA$ with preprocessing $p(|x|)\in O(|x|)$ and delay $d(|x|)\in O(1)$.
Note that we do not impose a restriction on the memory used. In particular, such an algorithm may use additional constant memory for writing between two consecutive answers.

Let $\pEnum{R_1}$ and $\pEnum{R_2}$ be enumeration problems.
There is an {\em exact reduction} from $\pEnum{R_1}$ to $\pEnum{R_2}$, denoted as $\pEnum{R_1}\leq_e\pEnum{R_2}$, if there exist mappings $\sigma$ and $\tau$ such that:
(1) for every $x\in\Sigma^*$, $\sigma(x)$ is computable in $O(|x|)$ time;
(2) for every $y$ such that $(\sigma(x),y)\in R_2$, $\tau(y)$ is computable in $O(1)$ time; and
(3) in multiset notation, $\{\tau(y)\mid (\sigma(x),y)\in R_2\} =
	\{y'\mid (x,y')\in R_1\}$.
Intuitively, $\sigma$ maps instances of $\pEnum{R_1}$ to instances of $\pEnum{R_2}$,
and $\tau$ maps solutions of $\pEnum{R_2}$ to solutions of $\pEnum{R_1}$.
If $\pEnum{R_1}\leq_e\pEnum{R_2}$ and $\pEnum{R_2}\in\DelayClin$, then $\pEnum{R_1}\in\DelayClin$ as well~\cite{bdg:dichotomy}.

\subsubsection*{Computational Hypotheses}
We use the following well-known hypotheses for lower bounds on certain computational problems:
\begin{description}
 \item[\matmul:] two Boolean $n\times n$ matrices cannot be multiplied in time $O(n^2)$.
 This problem is equivalent to the evaluation of the query
$\Pi(x,y)\leftarrow A(x,z),B(z,y)$ over the schema $\{A,B\}$ where $A,B\subseteq\{1,\ldots,n\}^2$.
It is conjectured that this problem cannot be solved in $O(n^2)$ time, 
and the best algorithms today require $O(n^\omega)$ time, where $2.37< \omega< 2.38$ is the matrix multiplication exponent~\cite{MatMulLimits}.
\\
Remark: \matmul{} is the assumption originally used by Bagan et al.~\cite{bdg:dichotomy} to show the hardness for CQs. It was later suggested to use a sparse variant of this hypothesis to show the same results~\cite{berkholz2020tutorial}. Here, we use the original variant, as the sparse variant is not enough to handle cases like Lemma~\ref{lemma:guardedTwoQueries}.
 \item[$\hyperclique$:] for all $k\geq 3$, finding a $k$-hyperclique in a $(k-1)$-uniform hypergraph with $n$ vertices is not possible in time $O(n^{k-1})$. When $k=3$, this is the assumption that we cannot detect a triangle in a graph in $O(n^2)$ time~\cite{cycles}. When $k>3$, this is a special case of the $(\ell,k)-$ Hyperclique Hypothesis~\cite{DBLP:conf/soda/LincolnWW18}, which states that, in a $k$-uniform hypergraph of $n$ vertices, $n^{k-o(1)}$ time is required to find a set of $\ell$ vertices such that each of it subsets of size $k$ forms a hyperedge.
 The $\hyperclique$ hypothesis is sometimes called Tetra$\langle k \rangle$~\cite{bb:thesis}.
 \item[\fourclique:] it is not possible to determine the existence of a $4$-clique in a graph with $n$ nodes in time $O(n^3)$.
 This is a special case of the $k$-Clique Hypothesis~\cite{DBLP:conf/soda/LincolnWW18}, which states that detecting a clique in a graph with $n$ nodes requires $n^{\frac{\omega k}{3}-o(1)}$ time, where $\omega$ is again the matrix multiplication exponent.
 \end{description}

\subsubsection*{Enumerating Answers to UCQs}
Given a UCQ $Q$ over some schema $\calS$, we denote by $\pEnum{Q}$ the enumeration problem $\pEnum{R}$, where $R$ is the binary relation between instances $I$ over $\calS$ and sets of mappings $Q(I)$.
We consider the size of the query as well as the size of the schema to be fixed. 
In the case of CQs, Bagan et al.~\cite{bdg:dichotomy} showed that a self-join-free acyclic CQ is in $\DelayClin$ iff it is free-connex. In addition, Brault-Baron~\cite{bb:thesis} showed that self-join-free cyclic queries are not in $\DelayClin$. In fact, the existence of a single answer to a cyclic CQ cannot be determined in linear time.
We call a CQ \emph{difficult} if it is self-join-free and not free-connex. Difficult CQs are intractable according to the following dichotomy.

\begin{theorem}[\cite{bdg:dichotomy,bb:thesis}]\label{theorem:originalDichotomy}
	Let $Q$ be a self-join-free CQ.
	\begin{enumerate}
		\item If $Q$ is free-connex, then $\pEnum{Q}\in\DelayClin$.
		\item If $Q$ is acyclic and not free-connex,  then $\pEnum{Q}\not\in\DelayClin$, assuming \matmul.
		\item If $Q$ is cyclic, then $\pEnum{Q}\not\in\DelayClin$, as $\qDecide{Q}$ cannot be solved in linear time, assuming $\hyperclique$.
	\end{enumerate}
\end{theorem}
The positive case of this dichotomy can be shown using the
\emph{Constant Delay Yannakakis} (\CDY) algorithm~\cite{CDY}. It uses an ext-$\free(Q)$-connex tree $T$ for $Q$. First, it performs the classical Yannakakis preprocessing~\cite{Yannakakis} over $T$ to obtain a relation for each node in $T$,
where all tuples can be used for some answer in $Q(I)$. Then, it considers only the subtree of $T$ containing $\free(Q)$, and joins the relations corresponding to this subtree with constant delay.

\section{Upper Bounds via Union Extensions}\label{sec:positive}

In this section, we identify tractable UCQs.
Section~\ref{sec:unions-of-tractable-CQs} discusses unions that contain only tractable CQs.
Section~\ref{sec:cheaters} inspects the requirements from UCQs in $\DelayClin$ and proves the Cheater's Lemma: a tool that allows us to compile several enumeration algorithms into one. 
In Section~\ref{sec:union-extensions}, we introduce the concepts of \emph{union extensions} and variable sets that a CQ can \emph{supply} in order to help the evaluation of another CQ in the union. We generalize the notion of free-connexity to UCQs and show that such queries
are in $\DelayClin$ in Section~\ref{sec:free-connex-UCQ-tractable}.

\subsection{Unions of Tractable CQs}\label{sec:unions-of-tractable-CQs}

Using known techniques~\cite[Proposition 2.38]{strozecki:thesis} a union of tractable CQs is also tractable. 

\begin{theorem}\label{thm:easyunions}
Let $Q=Q_1\cup\ldots\cup Q_n$ be a UCQ for some fixed $n\geq 1$. If all CQs in $Q$ are free-connex, then $\pEnum{Q}\in\DelayClin$.
\end{theorem}
\begin{proof}
Algorithm~\ref{algorithm:tractableUCQs} evaluates a union of two CQs.
In case of a union $Q=\bigcup_{i=1}^{\ell} Q_i$ of more CQs, we can use this recursively by treating the second query as $Q_2\cup\ldots\cup Q_\ell$.

\begin{algorithm}
\caption{Answering a union of two tractable CQs}
\label{algorithm:tractableUCQs}
\begin{algorithmic}[1]
\While{$a \gets Q_1(I).next()$}
\If {$a\not\in Q_2(I)$}
    	\State print $a$\label{algline:print1}
\Else
    \State print $Q_2(I).next()$\label{algline:print2}
\EndIf
\EndWhile
\While {$a \gets Q_2(I).next()$}
    	\State print $a$\label{algline:print3}
\EndWhile
\end{algorithmic}
\end{algorithm}

By the end of the run, the algorithm prints $Q_1(I)\setminus Q_2(I)$ over all iterations of line~\ref{algline:print1}, and it prints $Q_2(I)$ in lines~\ref{algline:print2} and~\ref{algline:print3}.
Line~\ref{algline:print2} is called $Q_1(I)\cap Q_2(I)$ times, so the command $Q_2(I).next()$ always succeeds there.
For free-connex CQs after linear preprocessing time, the answers can be enumerated in constant delay, and testing whether a given mapping is an answer can be done in constant time~\cite{berkholz2020tutorial}. Thus, this algorithm runs within the required time bounds.
\end{proof}

The technique presented in the proof of Theorem~\ref{thm:easyunions} has the advantage that it does not require more than constant memory available for writing in the enumeration phase.
Alternatively, this theorem is a consequence of the following lemma, which we prove next and gives us a general approach to compile several enumeration algorithms into one.

\subsection{The Cheater's Lemma}\label{sec:cheaters}

In this section, we prove a lemma that is useful to show upper bounds for UCQs even in cases not covered by Theorem~\ref{thm:easyunions}.
To get a clearer notion of what it means for a problem to be in the class $\DelayClin$, we first define linear partial time.

\begin{definition}[Linear Partial Time]\label{def:linpartial}
An algorithm runs in \emph{linear partial time} if, for every input $x$, the time before the $n$th output is $O(|x|+n)$.
\end{definition}

We claim next that if we relax the requirement of linear preprocessing and constant delay to allow a constant number of linear delay steps, we get linear partial time.

\begin{proposition}\label{prop:partial}
Let $\calA$ be an algorithm.
If there exist constants $a$, $b$ and $c$ such that, on any input $x$, the time between successive outputs of $\calA$ is bounded by $a(|x|+n)$ at most $b$ times, where $n$ is the number of answers printed up to that point, and bounded by $c$ otherwise,
then $\calA$ runs in linear partial time.
\end{proposition}
\begin{proof}
Before the $n$th answer, there are at most $b$ steps with delay larger than constant, and in these cases the delay is bounded by $a(|x|+n)$.
Thus, the time in which the $n$th answer is produced is bounded by
$ba(|x|+n)+cn\leq (ab+c)(|x|+n)$ for every $n$.
\end{proof}

If space is not restricted (as in our case), this relaxation in the phrasing of the requirements does not change the requirements themselves, as we show that any algorithm that runs in linear partial time can be modified to achieve linear preprocessing and constant delay. This can be done using the known technique~\cite[Proposition 12]{CAPELLI2018} of delaying the results to regularize the delay. In fact, we can further relax the phrasing by allowing a constant number of duplicates per answer.

\begin{lemma}[(The Cheater's Lemma)]\label{lemma:cheaters}
Let $P$ be an enumeration problem.
The following are equivalent:
\begin{enumerate}
    \item $P\in\DelayClin$.\label{cheater:delay}
    \item There exist an algorithm $\calA$ and constants $c$ and $d$ such that:
    $\calA$ outputs the solutions to $P$, the time before the $n$th answer is bounded by $c(|x|+n)$, and every result is produced at most $d$ times.\label{cheater:partial}
\end{enumerate}
\end{lemma}

\begin{proof}
As any algorithm that runs in linear preprocessing and constant delay also runs in linear partial time, the direction \ref{cheater:delay}$\Rightarrow$\ref{cheater:partial} is trivial. We now show the opposite direction.

We describe an algorithm $\calA'$ that simulates $\calA$, stores all generated results to prevent duplicates and holds back generated results to regularize the delay.
$\calA'$ maintains a lookup table containing the results that $\calA$ generated and a queue containing those that were not yet printed. Both are initialized as empty.
$\calA'$ calls $\calA$.
When $\calA$ returns a result, $\calA'$ checks the lookup table to determine whether it was found before. If it was not, the result is added to both the lookup table and the queue. Otherwise, it is ignored.
Since $\calA$ runs in linear partial time, there exists a constant $c$ such that the $n$th answer to $\calA$ is obtained after $c(|x|+n)$ operations.
$\calA'$ first performs $c|x|$ computation steps, and then after every $cd$ computation steps, it outputs a result from the queue.
The queue is never empty when used: $\calA'$ returns its $i$th result after $c|x|+(cd)i = c(|x|+di)$ computation steps; At this time, $\calA$ produced at least $di$ results, which contain at least $i$ unique results.  When it is done simulating $\calA'$, $\calA$ outputs all remaining results in the queue.
By definition, $\calA'$ operates with linear preprocessing time and constant delay. It outputs all results of $\calA$ with no duplicates since, due to the lookup table, every result enters the queue exactly once.
\end{proof}

The Cheater's Lemma formalizes how much we are allowed to ``cheat'' in order to show that a problem can be solved with linear preprocessing and constant delay by only showing an algorithm with relaxed requirements.
To show that a problem is in $\DelayClin$, it suffices to find an algorithm for this problem where the delay is usually constant, but it may be linear a constant number of times, and the number of times every result is produced is bounded by a constant.
The allowed linear delay is not only with respect to the input, but also with respect to the already produced answers.

\subsection{Union Extensions}\label{sec:union-extensions}

As Example~\ref{example:first} shows, Theorem~\ref{thm:easyunions} does not cover all tractable UCQs. We now define union extensions and address the other cases. We first define \emph{body-homomorphisms} between CQs to have the standard meaning of homomorphism, but without
the restriction on the heads of the queries.

\begin{definition}
Let $Q_1,Q_2$ be CQs.
\begin{itemize}
\item A \emph{body-homomorphism} from $Q_2$ to $Q_1$ is a mapping $h:\ivar(Q_2)\rightarrow\ivar(Q_1)$ such that for every atom $R(\vec{v})$ of $Q_2$, $R(h(\vec{v}))\in Q_1$.
\item If there exists a body-homomorphism $h$ from $Q_2$ to $Q_1$ and vice versa,
we say that $Q_1$ and $Q_2$ are \emph{body-homomorphically equivalent}.
\item A \emph{body-isomorphism} from $Q_2$ to $Q_1$ is a bijective mapping $h$ such that $h$ is a body-homomorphism from $Q_2$ to $Q_1$ and $h^{-1}$ is a body-homomorphism from $Q_1$ to $Q_2$.
\item If there is a body-isomorphism between $Q_2$ and $Q_1$, we say that $Q_1$ and $Q_2$ are \emph{body-isomorphic}.
\end{itemize}
\end{definition}

The standard definition of a homomorphism from $Q_2$ to $Q_1$ is then a body-homomorphism $h$ from $Q_2$ to $Q_1$ such that $h(\free(Q_2))=\free(Q_1)$.
To demonstrate the definition of a body-homomorphism, consider the UCQ $Q=Q_1\cup Q_2$ from Example~\ref{example:first} with the CQs $Q_1(x,y,w)\leftarrow R_1(x,z),R_2(z,y),R_3(y,w)$ and $Q_2(x,y,w)\leftarrow R_1(x,y),R_2(y,w)$.
There is a body-homomorphism $h:Q_2\rightarrow Q_1$ with $h(x)=x$, $h(y)=z$, and $h(w)=y$.
However, this is not a body-isomorphism.
Body-isomorphic CQs have the same body up to variable renaming (and the renaming is given by the body-isomorphism).
Note that body-isomorphic CQs are necessarily body-homomorphically equivalent.
The opposite direction holds for self-join-free CQs; that is, if two CQs are self-join-free and body-homomorphically equivalent, then they are necessarily body-isomorphic.
We now formalize the way that one CQ can help with evaluating another CQ in the union by supplying variables.

\begin{definition}\label{def:provides}
We say that a CQ $Q$ \emph{supplies} a set $V$ of variables if there exists $S$ such that $V\subseteq S \subseteq \free(Q)$ and $Q$ is $S$-connex.
\footnote{The conference version of this article contains a definition of \emph{providing} variables instead of that of \emph{supplying} variables. The difference is that the supplied variables are specified in terms of the giving CQ, while the providing variables are specified in terms of the receiving CQ. We need the terminology to refer to the giving CQ to be able to specify a variable set that one CQ supplies and \emph{several} CQs use, and so we introduced this minor change of naming with respect to the previous version in order to avoid confusion.}
\end{definition}

When considering again Example~\ref{example:first}, the query $Q_2$ supplies $\{x,y,w\}\subseteq\free(Q_2)$ as $Q_2$ is $\{x,y,w\}$-connex.
The following lemma shows why body-homomorphisms and supplying variables play an important role in UCQ enumeration.
In case of a body-homomorphism from a supplying CQ $Q_2$ to another CQ $Q_1$, we can produce an auxiliary relation that contains all possible value combinations of the matching variables in $Q_1$. This can be done efficiently while producing some answers to $Q_2$.
\footnote{The version of Lemma~\ref{lemma:provide} as it was phrased in the conference version of this article was incorrect. In fact, it was too strong, and the slightly weaker version given here describes more accurately what we need.}

\begin{lemma}\label{lemma:provide}
Let $Q_2$ be a CQ that supplies the variables $\vec{v}_2$.
Given an instance $I$,
one can compute with linear time preprocessing and constant delay a set of mappings $M$ from $\free(Q_2)$ to the domain such that:
\begin{itemize}
\item $M\subseteq Q_2(I)$
\item $M$ can be translated in time $O(|M|)$ to a relation $R^{M}$ such that:
for every CQ $Q_1$ with a body-homomorphism $h$ from $Q_2$ to $Q_1$
and for every answer $\mu_1\in\full(Q_1)(I)$,
there is the tuple $\mu_1(h(\vec{v}_2))\in R^{M}$.
\end{itemize}
\end{lemma}

\begin{proof}
According to Definition~\ref{def:provides}, there exists $\vec{v}_2\subseteq S \subseteq \free(Q_2)$ such that $Q_2$ is $S$-connex.
Take an ext-$S$-connex tree $T$ for $Q_2$, and perform the CDY algorithm~\cite{CDY} on $Q_2$ while treating $S$ as the free variables.
This results in a set $N$ of mappings from the variables of $S$ to the domain such that $N=Q_2(I)|_S$.
Note that, as we mentioned following Theorem~\ref{theorem:originalDichotomy}, the CDY algorithm has a preprocessing stage that removes dangling tuples and guarantees that, at its end, there is a relation for each vertex of the tree such that each tuple of such a relation can be used for some answer.

For every mapping $\mu\in N$, we extend it once to obtain a mapping from all variables of $Q_2$ as follows.
Go over all vertices of $T$ starting from the connected subtree containing $S$ and treating a neighbor of an already treated vertex at every step.
Consider a step where in its beginning $\mu$ is a homomorphism from a set $S_1$, and we are treating an atom $R(\vec{v},\vec{u})$ where  $\vec{v}\subseteq S_1$ and $\vec{u}\cap S_1 = \emptyset$. We take some tuple in $R$ of the form $(\mu(\vec{v}),\vec{t})$ and extend $\mu$ to also map $\mu(\vec{u})=\vec{t}$. Such a tuple exists since the dangling tuples were removed. This extension takes constant time, and in its end we have that $\mu|_{\free(Q_2)}\in Q_2(I)$. 
Using these extensions, we set $M=\{\mu^+|_{\free(Q_2)}\mid \mu^+\text{ is an extension of }\mu\in N\}$. We have that
$M\subseteq Q_2(I)$. We also have that $M|_S=N=Q_2(I)|_S$, and since $\vec{v}_2\subseteq S$, this means that $M|_{\vec{v}_2}=Q_2(I)|_{\vec{v}_2}$.
The mappings $M$ are computed with linear preprocessing and constant delay as this is the complexity of the CDY algorithm and the manipulations we describe only require constant time per answer.

We define $R^{M}=\{\mu^+(\vec{v}_2)\mid \mu^+\in M\}$.
Since $M|_{\vec{v}_2}=Q_2(I)|_{\vec{v}_2}$, this is the same as $\{\mu_2(\vec{v}_2)\mid \mu_2\in Q_2(I)\}$.
Let $Q_1$ be a CQ such that there is a body-homomorphism $h$ from $Q_2$ to $Q_1$, and
let $\mu_1\in \full(Q_1)(I)$.
Since $h$ is a body-homomorphism, for every atom $R(\vec{v})$ in $Q_2$, $R(h(\vec{v}))$ is an atom in $Q_1$. Since $\mu_1$ forms an answer to the full $Q_1$, for every such atom $\mu_1(h(\vec{v}))\in R^I$.
This means that $\mu_1\circ h|_{\free(Q_2)}$ is an answer to $Q_2$, so there exists $\mu_2\in Q_2(I)$ such that $\mu_1\circ h|_{\free(Q_2)}=\mu_2$.
By construction, $\mu_1(h(\vec{v}_2))=\mu_2(\vec{v}_2)\in R^{M}$.
\end{proof}
Note that if the mapping $h$ is not a body-homomorphism, Lemma~\ref{lemma:provide} does not hold in general.
\begin{example}
Consider $Q=Q_1\cup Q_2$, which is a slight modification of Example~\ref{example:first},
with
\begin{align*}
Q_1(x,y,w)&\leftarrow R_1(x,z),R_2(z,y),R_3(y,w)\text{ and }\\ 	
Q_2(x,y,w)&\leftarrow R_1(x,y),R_2(y,w), R_4(y).
\end{align*}
Since $R_4$ is not a relational symbol in $Q_1$, there is no body-homomorphism from
$Q_2$ to $Q_1$.
If, $R_4^I=\idom$, then we can take the
same approach as in Example~\ref{example:first}, as the answers of $Q_2$ form $Q_1(I)|_{\{x,z,y\}}$. However, if $R_4^I$ is smaller, this extra atom may filter the answers to $Q_2$, and we do not obtain all of $Q_1(I)|_{\{x,z,y\}}$ in general.
\qed
\end{example}

During evaluation, a set of supplied variables can form an auxiliary relation, accessible by an auxiliary atom. We call the query with its auxiliary atoms a \emph{union extension}.

\begin{definition}
Let $Q= Q_1\cup\ldots\cup Q_n$ be a UCQ.
An \emph{extension sequence} for $Q$ is a sequence $Q^1,\ldots,Q^N$ where $Q=Q^1$ and for all $1< j \le N$, $Q^{j}$ is a UCQ of the form $Q_1^{j}\cup\ldots\cup Q_n^{j}$ such that the following holds. For some relational symbol $R_j$ that does not appear in $Q_{j-1}$ and a sequence $\vec{v}_j$ of variables supplied by a previous $Q_{p(j)}^\ell$ (i.e.,  $\ell\le j$ and $1\le p(j)\le n$) we have the following for all
$i=1,\dots,n$:
\begin{itemize}
    \item $Q_i^{j}=Q_i^{j-1}$; or
    \item there is a body-homomorphism $h_{{p(j)},i}$ from $Q_{p(j)}^1$ to $Q_i^1$, and $Q_i^{j}$ is obtained by adding the atom $R_j(h_{{p(j)},i}(\vec{v}_j))$ to $Q_i^{j-1}$.
\end{itemize}
If such an extension sequence exists, we call $Q^N$ a \emph{union extension} of $Q^1$.
Atoms that appear in $Q^N$ but not in $Q^1$ are called \emph{virtual atoms}.
\end{definition}

Consider Example~\ref{example:first} again. We already established that $Q_2$ supplies $\{x,y,w\}\subseteq\free(Q_2)$ and that there is a body-homomorphism $h:\var(Q_2)\rightarrow\var(Q_1)$ 
with $h((x,y,w))=(x,z,y)$.
Thus, $Q$ has the union extension that contains $Q_2$ and 
$Q_1^+(x,y,w)\leftarrow R_1(x,z),R_2(z,y),R_3(y,w),R'(x,z,y)$, where $Q_1^+$ is obtained by adding the virtual atom $R'(x,z,y)$ to $Q_1$. 

\subsection{Extension-Based Tractability}\label{sec:free-connex-UCQ-tractable}

Union extensions can transform an intractable query to a free-connex one.
		  \begin{definition}
		  Let $Q= Q_1\cup\ldots\cup Q_n$ be a UCQ. 
		  \begin{itemize}
		\item $Q_1$ is said to be \emph{union-free-connex} with respect to $Q$ if it has a free-connex union extension.
		\item $Q$ is \emph{free-connex} if all CQs in $Q$ are union-free-connex.
	\end{itemize}
\end{definition}

Note that the term free-connex for UCQs is a generalization of that for CQs: If
a UCQ $Q$ contains only one CQ, then $Q$ is free-connex iff the CQ it contains is free-connex.
We next show that tractability of free-connex queries also carries over to UCQs.

\begin{theorem}\label{thm:positive}
Let $Q$ be a UCQ. If $Q$ is free-connex, then $\pEnum{Q}\in\DelayClin$.
\end{theorem}

\begin{proof}
We begin by sketching the proof. For answering $Q$, we describe an algorithm $\calA$ which is comprised of two phases: a provision phase and a final results phase.
In the provision phase, a free-connex union-extension is instantiated  for each CQ in the union.
In the final results phase, 
the answers of every CQ in the union are enumerated via their free-connex union-extensions.
These answers can be enumerated efficiently using the CDY algorithm (see Theorem~\ref{theorem:originalDichotomy}).
We will use Lemma~\ref{lemma:provide} to generate some answers to $Q$ during the provision phase; this permits us to use more than linear time before the final results phase while still achieving an enumeration algorithm with only linear preprocessing time.
We will use the Cheater's Lemma (Lemma~\ref{lemma:cheaters}) to remove duplicates and obtain the time bounds we want.

\paragraph{Initial Algorithm.}
Let $Q=Q_1\cup\ldots\cup Q_n$, and take an extension sequence $Q^1,\ldots,Q^N$ where $Q^1=Q$ and $Q^N$ comprises of free-connex CQs.
The provision phase consists of $N-1$ provision steps.
Let $I=I^1$ be the input database instance.
During the $j$th provision step (with $1<j\le N$), 
we extend the database instance $I_{j-1}$ into an instance $I_j$ that matches $Q^j$. 
We use Lemma~\ref{lemma:provide} to generate a set of answers
$M_j\subseteq Q_{p(j)}^{j-1}(I_{j-1})$ while also computing a relation $R^{M_j}$.
We set $(R_j)^{I_j}:=R^{M_j}$. The other relations remain as they were; that is, $R^{I_{j}}=R^{I_{j-1}}$ for every relational symbol in the UCQ except for $R_j$.
By the end of the provision phase, the algorithm computes an instance $I_N$ that matches $Q^N$.
Finally, we perform the final results phase where we compute $Q_i^N(I_N)$ for every $Q_i$ in the union using the CDY algorithm.

\paragraph{Correctness.}
We prove that for all $1\le i\le n$ and $1\leq j\leq N$ we have that $Q_i(I)=Q_i^{j}(I_j)$ by induction on $j$. 
The base case trivially holds as $Q_i(I)=Q_i^{1}(I_1)$ by definition.
Now consider the $j$th provision step.
Intuitively, answers to an extended CQ are by definition all answers to its previous version that agree with some tuple in the new atom. Since the new atom contains a projection of the answers, the extension has exactly the same answers as its previous version.
More formally, let $Q_{e(j)}$ be a CQ extended in step $j$.
For every mapping $\mu$, by definition of the extension we have that 
$\mu\in Q_{e(j)}^{j}(I_j)$ if and only if both $\mu\in Q_{e(j)}^{j-1}(I_{j-1})$ and $\mu(h(\vec{v}_j))\in R_j^{I_j}$, and these two conditions hold if and only if $\mu\in Q_{e(j)}^{j-1}(I_{j-1})$ since for all $\mu\in Q_{e(j)}^{j-1}(I_{j-1})$ we have that $\mu(h(\vec{v}_j))\in R_j^{I_j}$. This proves that $Q_{e(j)}^{j}(I_{j})=Q_{e(j)}^{j-1}(I_{j-1})$.
This shows that for every $Q_i$ in the union, $Q_i^{j}(I_{j})=Q_i^{j-1}(I_{j-1})$. By the induction hypothesis, $Q_i^{j-1}(I_{j-1})=Q_i(I)$.
This concludes the proof that $Q_i^{j}(I_{j})=Q_i(I)$ for all $j\leq N$, and in particular, $Q^N(I_N)=Q(I)$.
Since $Q^N(I_N)=Q(I)$, the final results phase computes all answers to $Q$.
Since $Q_i(I)=Q_i^{j}(I_j)$ for all $i$ and $j$, the mappings generated by Lemma~\ref{lemma:provide} during the provision phase are also answers to $Q$.

\medskip
The algorithm $\calA$ we presented so far produces the results we want, but it is not a constant delay enumeration algorithm.  
It is left to show that $\calA$ conforms to the conditions of the Cheater's lemma. This would mean that we can apply the lemma, and conclude that $\pEnum{Q}\in\DelayClin$.

\paragraph{Duplicates.}
Overall the algorithm produces results in $N$ provision steps during the provision phase and in $n$ CQ evaluation steps during the final results phase.
The results produced at each individual step contain no duplicates. Therefore, every result appears at most $N+n$ times. Since $N$ and $n$ are constants, this is a constant number of duplicates per value.

\paragraph{Delay.}
Consider the $j$th provision step. It starts with a preprocessing of time $\calO(|I_{j-1}|)$ followed by a constant delay enumeration of a set $M_j$ of answers. At its end, $\calO(|M_j|)$ time is required to compute the new relation.
This means that the delay before the first answer of the $j$th provision step is $O(|I_{j-1}|+|M_{j-1}|)$, and the delay before answers in the provision phase that are not first in their step is constant.
During the final results phase, we have $n$ steps in which we apply the CDY algorithm on an extended CQ. Each such step starts with $O(|I_N|)$ preprocessing time followed by constant delay. Thus, during the final results phase, there are $n$ times where the delay is $O(|I_N|)$, and in other times the delay is constant.

We have that $|I_{j}|=|I_{j-1}|+|R_j^{M_j}|\leq |I_{j-1}|+|M_j|$ for all $1<j\le N$, and $|I_1|=|I|$.
By induction, this shows that $|I_{j}|\leq |I|+\sum_{i=2}^{j}{|M_i|}$.
This means that whenever the delay is not constant, it is linear in the input size plus the number of (not necessarily unique) answers produces thus far.
Since there are $n+N$ such times where the delay is not constant,
according to Proposition~\ref{prop:partial}, the algorithm $\calA$ runs in linear partial time, and the requirements of the Cheater's Lemma on the delay are met.
\end{proof}

We can now conclude the tractability of Example~\ref{example:first}.
In Section~\ref{sec:union-extensions}, we saw that $Q_1$ has the union extension
$Q_1^+(x,y,w)\leftarrow R_1(x,z),R_2(z,y),R_3(y,w),R'(x,z,y)$.
Since $Q_1^+$ is free-connex (see Figure~\ref{figure:ex1}), 
the CQ $Q_1$ is union-free-connex. Since every query in $Q$
is union-free-connex, we have that $\pEnum{Q}\in\DelayClin$ by Theorem~\ref{thm:positive}.

\begin{figure}
\begin{tikzpicture}
[
    he/.style={draw, rounded corners,inner sep=0pt},        
    ce/.style={draw,dashed, rounded corners=2pt}, 
]

\node (wy) at (0,0) {$w, y$};
\node[he, fit = (wy)] {};					 

\node (xy) at (0,1) {$x,y$};
\node[he, fit = (xy)] {};					 				
								 
\node[ce, fit = (wy) (xy)] {};								 
								 
\draw (wy) -- (xy);								 	


\node (xzw) at (4.3,0) {$x,z,y$};
\node[he, fit = (xzw)] {};

\node (xy2) at (5.6,0) {$x,y$};
 \node[he, fit = (xy2)] {};

\node (wy2) at (5.6,1) {$w,y$};
 \node[he, fit = (wy2)] {};

\node (xz) at (3,0) {$x,z$};
\node[he, fit = (xz)] {};

\node (yz) at (4.3,1) {$y,z$};
\node[he, fit = (yz)] {};

\node[ce, fit = (xy2) (wy2)] {};

\draw (xzw) -- (yz);								 	
\draw (xz) -- (xzw);								 	
\draw (xzw) -- (xy2);								 	
\draw (xy2) -- (wy2);								 	

\node (Q2) at (-1,0.5) {$Q_2:$};
\node (Q2plus) at (2,0.5) {$Q_1^+$:};

\end{tikzpicture}
\caption{A $\{x,y,w\}$-connex tree for $Q_2$ and a $\{x,y,w\}$-connex tree for $Q_1^+$ of Example~\ref{example:first}.}
\label{figure:ex1}
\end{figure}

\begin{remark}\em
Example~\ref{example:first} is a counter example to a past made claim~\cite[Theorem 4.2b]{DBLP:conf/icdt/BerkholzKS18}. 
The claim is that if a UCQ contains an intractable CQ and does not contain redundant CQs (a CQ contained in another CQ in the union), then the union is intractable.
In contrast, none of the CQs in Example~\ref{example:first} is redundant, $Q_1$ is intractable, and yet the UCQ is tractable.

The intuition behind the proof of the past claim is reducing the hard CQ $Q_1$ to $Q$. This can be done by assigning each variable of $Q_1$ with a different and disjoint domain (e.g., by concatenating the variable names to the values in the relations corresponding to the atoms), and leaving the relations that do not appear in the atoms of $Q_1$ empty. It is well known that $Q_1\subseteq Q_2$ iff there exists a homomorphism from $Q_2$ to $Q_1$~\cite{homomorphism/redundancy}. The claim is that since there is no homomorphism from another CQ in the union to $Q_1$, then there are no answers to the other CQs with this reduction. However, it is possible that there is a body-homomorphism from another CQ to $Q_1$ even if it is not a full homomorphism (the free variables do not map to each other). Therefore, in cases of a body-homomorphism, the reduction from $Q_1$ to $Q$ does not work. In such cases, the union may be tractable, as we show in Theorem~\ref{thm:positive}. In Lemma~\ref{lemma:body-homo}, we use the same proof described here, but restrict it to UCQs where there is no body-homomorphism from other CQs to $Q_1$.
\qed
\end{remark}

The tractability result in Theorem~\ref{thm:positive} is based on the structure
of the union-extended queries. This means that the intractability of any query within a UCQ
can be resolved as long as another query can supply the right variables. The following example
shows that this can even be the case for a UCQ only consisting of non-free-connex CQs. It also illustrates 
why the definition of union extensions needs to be recursive.
\begin{example}\label{example:yellow}
Let $Q=Q_1\cup Q_2\cup Q_3$ with
\begin{align*}
Q_1(x,y,v,u)\leftarrow &R_1(x,z_1),R_2(z_1,z_2),R_3(z_2,z_3),R_4(z_3,y),R_5(y,v,u),\\
Q_2(x,y,v,u)\leftarrow &R_1(x,y),R_2(y,v),R_3(v,z_1),R_4(z_1,u),R_5(u,t_1,t_2),\\
Q_3(x,y,v,u)\leftarrow &R_1(x,z_1),R_2(z_1,y),R_3(y,v),R_4(v,u),R_5(u,t_1,t_2).
\end{align*}
Each of three CQs is difficult on its own: $Q_1$ has the free-path $(x,z_1,z_2,z_3,y)$, while $Q_2$ has the free-path $(v,z_1,u)$, and $Q_3$ has the free-path $(x,z_1,y)$.
The CQ $Q_2$ supplies $\{x,y,v\}$ as $Q_2$ is $\{x,y,v\}$-connex and
$\{x,y,v\}$ are free in $Q_2$.
Since there is a body-homomorphism $h_{2,3}$ from $Q_2$ to $Q_3$ with
$h_{2,3}((x,y,v))=(x,z_1,y)$, we can 
extend the body of $Q_3$ by the virtual atom $R'(x,z_1,y)$, which yields a free-connex extension $Q_3^+$ of $Q_3$.
Similarly, we have that $Q_3$ supplies $\{y,v,u\}$, and there is a body-homomorphism $h_{3,2}$ from $Q_3$ to $Q_2$ with
$h_{3,2}((y,v,u))=(v,z_1,u)$.
Extending $Q_2$ by $R''(v,z_1,u)$ yields the free-connex extension $Q_2^+$.
Since $Q_2^+$ and $Q_3^+$ each supply $\{x,y,v,u\}$, we can add virtual atoms with the variables $(x,z_1,z_2,y)$ and $(x,z_2,z_3,y)$ to $Q_1$.
This results in a free-connex extension $Q_1^+$.
The UCQ $Q_1^+\cup Q_2^+\cup Q_3^+$ is a free-connex union extension of $Q$.
By Theorem~\ref{thm:positive}, $\pEnum{Q}\in\DelayClin$.
\qed
\end{example}

\section{Lower Bounds}\label{sec:lower-bounds}

In this section, we prove lower bounds for evaluating UCQs within the time bounds of $\DelayClin$.
Since for CQs we currently only have hardness result when they are self-join free, we focus on unions of self-join-free CQs.
We begin with some general observations regarding cases where a UCQ is at least as hard as a single CQ it contains, and then continue to handle other cases.
In Section~\ref{sec:intractables} we discuss unions containing only difficult CQs, and in Section~\ref{sec:2-body-iso} we discuss unions containing two body-isomorphic CQs. In both cases such UCQs may be tractable, and in case of such a union of size two, we show that our results from Section~\ref{sec:positive} capture all tractable unions.

In order to provide some intuition for the choices we make throughout this section, we first explain where the approach used for proving the hardness of single CQs fails.
Consider Example~\ref{example:first}.
The original proof that shows that $Q_1$ is hard describes a reduction from Boolean matrix multiplication~\cite[Lemma 26]{bdg:dichotomy}.
Let $A$ and $B$ be binary representations of
Boolean $n\times n$ matrices, i.e. $(a,b)\in A$ corresponds to a $1$ in the first matrix at index $(a,b)$. Define 
a database instance $I$ as $R_1^I = A$, $R_2^I=B$, and $R_3^I=\{1,\ldots,n\}\times \{\bot\}$.
One can show that $Q_1(I)$ corresponds to the answers of $AB$.
If $\pEnum{Q_1}\in\DelayClin$, we can solve matrix multiplication in time $\calO(n^2)$, in contradiction to \matmul.
Since $Q_2$ evaluates over the same relations, $Q_2$ also produces answers over this construction. Since the number of results for $Q_2$ might reach up to $n^3$, evaluating $Q$ in constant delay does not necessarily compute the answers to $Q_1$ in $O(n^2)$ time, and does not contradict the complexity assumption.

So in general, whenever we show a lower bound to a UCQ by computing a hard problem through answering one CQ in the union, we need to ensure that the other CQs cannot have too many answers over this construction.
The following lemma formalizes the idea that by assigning variables of a CQ with different domains, we can restrict the answers obtained by other CQs.

\begin{lemma}\label{lemma:diff-domain}
Given a CQ $Q_1$ over a schema $\calS$, there exist mappings $\sigdiff$ and $\taudiff$ such that for every database instance $I$ over $\calS$:
\begin{itemize}
\item $Q_1(I)=\taudiff(Q_1(\sigdiff(I)))$ in multiset notation.
\item $\sigdiff(I)$ can be computed in linear time.
\item For every CQ $Q_i$ over $\calS$:
\begin{itemize}
\item $\taudiff(\mu)$ can be computed in constant time for every answer $\mu\in Q_i(\sigdiff(I))$.
\item If there is no body-homomorphism from $Q_i$ to $Q_1$, $Q_i(\sigdiff(I))=\emptyset$.
\item If there is no homomorphism from $Q_i$ to $Q_1$, given $\mu\in Q_1(\sigdiff(I))\cup Q_i(\sigdiff(I))$, it is possible to determine in constant time whether $\mu\in Q_1(\sigdiff(I))$.
\end{itemize}
\end{itemize}
\end{lemma}
\begin{proof}
We define $\sigdiff$ to assign each variable of $Q_1$ with a different and disjoint domain by concatenating the variable names to the values in their corresponding relations.
For every atom $R(v_1,\cdots,v_m)$ in $Q_1$ and tuple $(c_1,\cdots,c_m)\in R^I$, we add the tuple $((c_1,v_1),\ldots,(c_m,v_m))$ to  $R^{\sigdiff(I)}$. All relations that do not appear in $Q_1$ are left empty.
We claim that the results of $Q_1$ over the original instance are exactly the same as over our construction if we omit the variable names. That is, we define $\taudiff : \idom\times\ivar(Q_1) \rightarrow \idom$ as $\taudiff((c,v))=c$, and show that $Q_1(I)=\taudiff(Q(\sigdiff(I)))$. Note the $\sigdiff$ and $\taudiff$ can be computed in linear and constant time respectively.

We first prove that $Q_1(I)=\taudiff(Q_1(\sigdiff(I)))$.
The first direction is trivial:
if $\nu |_{\free(Q_1)}\in Q_1(\sigdiff(I))$, then for every atom $R(\vec{v})$ in $Q_1$, $\nu(\vec{v})\in R^{\sigdiff(I)}$. By construction, $\taudiff(\nu(\vec{v}))\in R^{I}$, and therefore $\taudiff\circ\nu|_{\free(Q_1)}\in Q_1(I)$.
We now show the opposite direction.
If $\mu |_{\free(Q_1)}\in Q_1(I)$, then for every atom $R(v_1,\cdots,v_m)$ in $Q_1$, $(\mu(v_1),\ldots,\mu(v_m))\in R^I$. By construction, $((\mu(v_1),v_1),\ldots,(\mu(v_m),v_m))\in R^{\sigdiff(I)}$. By defining $f_\mu: \ivar(Q_1) \rightarrow\idom\times\ivar(Q_1)$ as $f_\mu(u)=(\mu(u),u)$,
we have $f_\mu\in Q_1(\sigdiff(I))$. Since $\taudiff\circ f_\mu = \mu$, we have that $\mu |_{\free(Q_1)}\in \taudiff(Q_1(\sigdiff(I)))$, and this concludes that $Q_1(I)\subseteq\taudiff(Q_1(\sigdiff(I)))$.

We now show that $Q_i(\sigdiff(I))=\emptyset$ if there is no body-homomorphism from $Q_i$ to $Q_1$.
Assume by contradiction that there exists such $\mu |_{\free(Q_i)}\in Q_i(\sigdiff(I))$.
This means that for every atom $R(\vec{v})$ in $Q_i$, $\mu(\vec{v})\in R^{\sigdiff(I)}$.
By construction, $\mu(\vec{v})$, like all tuples in $R^{\sigdiff(I)}$, is of the form $((c_1,v_1),\ldots,(c_m,v_m))$ such that $R(v_1,\ldots,v_m)$ is an atom in $Q_1$.
Define $\eta: \idom \times \ivar(Q_1) \rightarrow\ivar(Q_1)$ as $\eta(c,v)=v$.
We have that for every atom $R(\vec{v})$ in $Q_i$,  $R(\eta(\mu(\vec{v})))$ is an atom  in $Q_1$. This means that $\eta\circ\mu$ is a body-homomorphism from $Q_i$ to $Q_1$, which is a contradiction.

It is left to show that, if there is no homomorphism from $Q_i$ to $Q_1$, given $\mu\in Q_1(\sigdiff(I))\cup Q_i(\sigdiff(I))$, it is possible to determine in constant time whether $\mu\in Q_1(\sigdiff(I))$. As we showed, $Q_i$ has answers over $\sigdiff(I)$ only if there is a body-homomorphism $h$ from it to $Q_1$. Since this is not a full homomorphism, $h(\free(Q_i))\neq\free(Q_1)$.
Note that by definition of UCQs, the free variables are the same for all CQs in the union, and in particular $\free(Q_i)=\free(Q_1)$.
If $\mu$ is an answer to $Q_1$, then applying $\eta\circ\mu$ on $\free(Q_1)$ will result in $\free(Q_1)$.
Otherwise (if $\mu$ is an answer to $Q_i$), then applying $\eta\circ\mu$ on $\free(Q_1)=\free(Q_i)$ will result in $h(\free(Q_i))$.
Since $h(\free(Q_i))\neq\free(Q_1)$, this helps us distinguish the answers:
apply $\eta\circ\mu$ on the $\free(Q_1)$; if this results in $\free(Q_1)$, then $\mu$ is an answer to $Q_1$; otherwise, $\mu$ is an answer to $Q_i$.
\end{proof}

\begin{example}
Consider the union of $Q_1(x)\leftarrow R(x,y)$ with $Q_2(x)\leftarrow R(y,x)$ and $Q_3(x)\leftarrow R(x,x)$.
If the input database $I$ has $R^I=\{(a,a),(a,b)\}$, then all three queries have $a$ as an answer, and $Q_2$ also has $b$.
Applying the construction from Lemma~\ref{lemma:diff-domain} results in $R^{\sigdiff(I)}=\{((a,x),(a,y)),((a,x),(b,y))\}$.
Then, we have that $Q_3(\sigdiff(I))=\emptyset$ as there is no body-homomorphism from $Q_3$ to $Q_1$.
There is a body-homomorphism from $Q_2$ to $Q_1$ but no such full homomorphism.
We have that $Q_2(\sigdiff(I))=\{(a,y),(b,y)\}$ and $Q_1(\sigdiff(I))=\{(a,x)\}$, and so the answers can be distinguished by the concatenated variable name.
Thus, if we are interested only in answers to $Q_1$ but we only have the ability to compute the entire UCQ, we can use this construction, ignore the answers containing $y$, and project out the variable names. This will result in  the answer $a$ alone.
\end{example}

The following lemma identifies cases where we can encode any arbitrary instance of a CQ to an instance of the union containing it, such that no other CQ in the union returns results.

\begin{lemma}\label{lemma:body-homo}
Let $Q$ be a UCQ, and let $Q_1\in Q$ such that for all $Q_i\in Q\setminus \{Q_1\}$ there is no body-homomorphism from $Q_i$ to $Q_1$.
Then $\pEnum{Q_1}\leq_e\pEnum{Q}$.
\end{lemma}
\begin{proof}
This reduction can be performed using the mappings $\sigdiff$ and $\taudiff$ defined in Lemma~\ref{lemma:diff-domain}.
For every $Q_i\in Q\setminus \{Q_1\}$, there is no body-homomorphism from $Q_i$ to $Q_1$, and therefore $Q_i(\sigdiff(I))=\emptyset$.
We now have that $\taudiff(Q(\sigdiff(I)))=\bigcup_{Q_i\in Q}{\taudiff(Q_i(\sigdiff(I)))}=\taudiff(Q_1(\sigdiff(I)))$. Since we also know that $\taudiff(Q_1(\sigdiff(I)))=Q_1(I)$, this concludes our reduction.
\end{proof}

The lemma above implies that if there is an intractable CQ in a union where no other CQ maps to it via a body-homomorphism, then the entire union is intractable. This also captures cases such as a union of CQs where one of them is hard, and
the others contain a relation that does not appear in the hard CQ.

Using the same reduction, a similar statement with relaxed requirements can be made in case it is sufficient to consider the decision problem.

\begin{lemma}\label{lemma:decision-reduction}
Let $Q$ be a UCQ, and let $Q_1\in Q$ such that for all $Q_i\in Q$,
 either there exists no body-homomorphism from $Q_i$ to $Q_1$,
 or $Q_1$ and $Q_i$ are body-isomorphic.
 Then $\qDecide{Q_1}\leq\qDecide{Q}$ via a linear-time many-one reduction.
\end{lemma}
\begin{proof}
We use the encoding from Lemma~\ref{lemma:diff-domain}.
We know that $Q_i(\sigdiff(I))=\emptyset$ for every CQ $Q_i$ with no body-homomorphism to $Q_1$. We claim now that a CQ $Q_j$ which is body-isomorphic to $Q_1$ has an answer iff $Q_1$ has an answer. Therefore $Q(\sigdiff(I))\neq\emptyset$ iff $Q_1(\sigdiff(I))\neq\emptyset$.
We know that $Q_1$ retains the same answers under this encoding, so in particular $Q_1(\sigdiff(I))=\emptyset$ if and only if $Q_1(I)=\emptyset$.
Overall this shows that $Q(\sigdiff(I))\neq\emptyset$ if and only if $Q_1(I)\neq\emptyset$.

We now show formally that for body-isomorphic CQs $Q_1$ and $Q_2$ and database $I$,
$Q_1(I)\neq\emptyset$ iff $Q_2(I)\neq\emptyset$.
Let $h$ be a body-homomorphism from $Q_2$ to $Q_1$. That is, for every atom $R(\vec{v})\in Q_2$, we have $R(h(\vec{v}))\in Q_1$.
If $Q_1(I)\neq\emptyset$, then there exists $\mu|_{\free(Q_1)}\in Q_1(I)$, and for every atom $R(h(\vec{v}))\in Q_1$,  we have 
$\mu(h(\vec{v}))\in R^I$.
This means that $\mu\circ h$ is a homomorphism from $Q_2$ to $I$, and $\mu\circ h|_{\free(Q_2)}\in Q_2(I)$. So, $Q_2(I)\neq\emptyset$.
Since there is also a body-homomorphism from $Q_1$ to $Q_2$, we can show in the same way that if $Q_2(I)\neq\emptyset$ then $Q_1(I)\neq\emptyset$.
\end{proof}

From Theorem~\ref{theorem:originalDichotomy} it follows that deciding whether a cyclic CQ has any answers cannot be done in linear time (assuming \hyperclique). 
Following Lemma~\ref{lemma:decision-reduction}, if a UCQ $Q$ contains a cyclic $Q_1$, and the conditions of Lemma~\ref{lemma:decision-reduction} are satisfied with respect to this CQ, then the entire union cannot be decided in linear time,
and thus $\pEnum{Q}\not\in\DelayClin$.

\subsection{Unions of Intractable CQs}\label{sec:intractables}

We now discuss unions containing only CQs classified as hard according to Theorem~\ref{theorem:originalDichotomy}. Recall that these are called \emph{difficult CQs}, and they are self-join-free CQs and not free-connex.
The following lemma can be used to identify a CQ on which we can apply Lemma~\ref{lemma:body-homo} or Lemma~\ref{lemma:decision-reduction}. 

\begin{lemma}\label{lemma:qinhards}
	Let $Q$ be a UCQ. There exists a query $Q_1\in Q$ such that for all $Q_i\in Q$ either there is no body-homomorphism from $Q_i$ to $Q_1$ or $Q_1$ and $Q_i$ are 
body-homomorphically equivalent.
\end{lemma}
\begin{proof}
Consider a longest sequence $(Q^1,\ldots,Q^m)$ of CQs from $Q$ such that for every $2\leq j \leq m$ there is a body-homomorphism from $Q^{j}$ to $Q^{j-1}$, but no body-homomorphism in the opposite direction. We claim that $Q_1=Q^m$ satisfies the conditions of the lemma.

First we show that such a sequence exists. We denote the body-homomorphism from $Q^{j}$ to $Q^{j-1}$ by $\mu^j$.
It is not possible that the same query appears twice in the sequence: if $Q^i=Q^j$ where $j>i$, then there is a mapping $\mu^{i+2}\circ\ldots\circ\mu^{j}$ from $Q^j=Q^i$ to $Q^{i+1}$, in contradiction to the definition of the sequence. Therefore, $m\leq|Q|$, and a longest sequence exists.
We now show that $Q^m$ satisfies the requirements.
First consider some $Q^j\in\{Q^1,\ldots,Q^{m-1}\}$.  There is a body-homomorphism from $Q^m$ to $Q^j$ which is the composition $\mu^{j+1}\circ\ldots\circ\mu^m$. Therefore, either there is no body-homomorphism from $Q^j$ to $Q^m$, or $Q^m$ and $Q^j$ are body-homomorphically equivalent.
In addition, $Q^m$ is body-homomorphically equivalent to itself with the identity mapping serving as the body-homomorphisms.
It is left to consider CQs that are not on the sequence. Let $Q_i\in Q\setminus\{Q^1,\ldots,Q^m\}$. If there is no body-homomorphism from $Q_i$ to $Q^m$, then we are done. Otherwise, if there is also no body-homomorphism from $Q_1$ to $Q^m$, then $(Q^1,\ldots,Q^m,Q_i)$ is a longer sequence, contradicting the maximality of the sequence we started with. Therefore, in this case $Q^m$ and $Q_i$ are body-homomorphically equivalent.
Overall we showed that for each query in $Q$ either there is no body-homomorphism from this query to $Q^m$ or these two CQs are body-homomorphically equivalent.
\end{proof}

Using the results obtained so far, we deduce a characterization of all cases of a union of
difficult CQs, except those that contain a pair of body-isomorphic acyclic CQs.

\begin{theorem}\label{theorem:union_of_intractable}
Let $Q$ be a UCQ of difficult CQs that does not contain two body-isomorphic acyclic CQs. Then, $Q\not\in\DelayClin$, assuming \matmul~and \hyperclique.
\end{theorem}
\begin{proof}
Let $Q_1$ be a CQ in $Q$ given by Lemma~\ref{lemma:qinhards}. 
By definition, difficult CQs are self-join-free, and so if two CQs in $Q$ are body-homomorphically equivalent, they are also body-isomorphic.
We treat the two possible cases of the structure of $Q_1$. 
In case $Q_1$ is acyclic, since we know that $Q$ does not contain body-isomorphic acyclic CQs, then for all $Q_i\in Q\setminus\{Q_1\}$ there is no body-homomorphism from $Q_i$ to $Q_1$. According to Lemma~\ref{lemma:body-homo}, $\pEnum{Q_1}\leq_e\pEnum{Q}$.
Since $Q_1$ is self-join-free acyclic non-free-connex, we have that $\pEnum{Q_1}\not\in\DelayClin$ assuming \matmul. Therefore $\pEnum{Q}$ is not in $\DelayClin$ either.
In case $Q_1$ is cyclic,
we use Lemma~\ref{lemma:decision-reduction} to conclude that $\qDecide{Q_1}\leq\qDecide{Q}$.
According to Theorem~\ref{theorem:originalDichotomy}, since $Q_1$ is self-join-free cyclic, $\qDecide{Q_1}$ cannot be solved in linear time assuming \hyperclique. Therefore $\qDecide{Q}$ cannot be solved in linear time either, and $\pEnum{Q}\not\in\DelayClin$.
\end{proof}

In the next example, we demonstrate how the reductions described in Lemma~\ref{lemma:decision-reduction} and Theorem~\ref{theorem:originalDichotomy} combine in Theorem~\ref{theorem:union_of_intractable}.

\begin{example}\label{example:intractables}
Consider the UCQ $Q=Q_1\cup Q_2\cup Q_3$ with
\begin{align*}
	Q_1(x,y)&\leftarrow R_1(x,y),R_2(y,u),R_3(x,u),\\
	Q_2(x,y)&\leftarrow R_1(y,v),R_2(v,x),R_3(y,x),\\
	Q_3(x,y)&\leftarrow R_1(x,z),R_2(y,z).
	\end{align*}
The queries $Q_1$ and $Q_2$ are cyclic, and $Q_3$ is acyclic but not free-connex.
This union is intractable according to Theorem~\ref{theorem:union_of_intractable}.
Note that $Q_1$ and $Q_2$ are body-isomorphic, but there is no body-homomorphism from $Q_3$ to $Q_1$.
The proof of Theorem~\ref{theorem:originalDichotomy} states the following:
If $\pEnum{Q_1}\in\DelayClin$, then given an input graph $G$, we can use $Q_1$ to decide the existence of triangles in $G$ in time $O(n^2)$, in contradiction to \hyperclique.
The same holds true for $\pEnum{Q}$.
For every edge $(u,v)$ in $G$ with $u<v$ we add $((u,x),(v,y))$ to $R_1^I$,  $((u,y),(v,z))$ to $R_2^I$ and $((u,x),(v,z))$ to $R_3^I$.
The query detects triangles: for every triangle $a,b,c$ in $G$ with $a<b<c$, the query $Q_1$ returns $((a,x),(b,y))$.
The union only returns answers corresponding to triangles:
\begin{itemize}
\item For every answer $((d,x),(e,y))$ to $Q_1$, there exists some $f$ such that $d,e,f$ is a triangle in $G$ with $d<e<f$.
\item For every answer $((g,z),(h,x))$ to $Q_2$, there exists some $i$ such that $g,h,i$ is a triangle in $G$ with $h<i<g$.
\item The query $Q_3$ returns no answers over $I$.
\qed
\end{itemize}
\end{example}

Theorem~\ref{theorem:union_of_intractable} 
does not cover the case of a UCQ containing acyclic non-free-connex queries with isomorphic  bodies. Since this requires a more intricate analysis, we first restrict ourselves to such unions of size two.

\subsection{Unions of Two Body-Isomorphic CQs}\label{sec:2-body-iso}

Consider a set of body-isomorphic CQs. As all of them have the same structure, either 
every CQ in this set is cyclic, or every CQ is acyclic.
In the case of a union of two cyclic CQs, the UCQ is intractable according to Theorem~\ref{theorem:union_of_intractable}.
So in this section, we discuss the union of body-isomorphic acyclic CQs. 
Note that, unlike the previous section, we allow a CQ in the union to be free-connex.
We first introduce a new notation for body-isomorphic UCQs that we use hereafter in order to simplify the discussion on the query structure.

Consider a UCQ of the form $Q_1\cup Q_2$,
where 
there exists a body-isomorphism
$h$ from $Q_2$ to $Q_1$. That is, the CQs have the structure:
	\begin{align*}
	Q_1(\vec{v}_1)&\leftarrow R_1(h(\vec{w}_1)),\ldots,R_n(h(\vec{w}_n)),\\
	Q_2(\vec{v}_2)&\leftarrow R_1(\vec{w}_1),\ldots,R_n(\vec{w}_n).
	\end{align*}
By applying $h^{-1}$ to the variables of $Q_1$, we can rewrite $Q_1$ as
$Q_1(h^{-1}(\vec{v}_1))\leftarrow R_1(\vec{w}_1),\ldots,R_n(\vec{w}_n)$.
Every mapping $\mu$ that is an answer to $Q_1$ in its new form can easily be converted to an answer to the original $Q_1$ by composing it with $h$. So for simplicity we can assume that $Q_1$ is given in the modified form.
Since now the two CQs have exactly the same body, we can treat the UCQ as a query with one body and two heads:
$$Q_1(h^{-1}(\vec{v}_1)),Q_2(\vec{v}_2)\leftarrow R_1(\vec{w}_1),\ldots,R_n(\vec{w}_n)$$
We use this notation from now on for UCQs containing only body-isomorphic CQs.
Note that when treating a UCQ as one CQ with several heads, we can use the notation $\atoms(Q)$, 
as the atoms are the same for all CQs in the union, and the notation $\free(Q_i)$,
as the free variables may differ between different queries $Q_i$ in the union.
With this notation at hand, we now inspect some examples of two body-isomorphic acyclic CQs.

\begin{example}\label{example:body-iso-matmul}
Consider $Q=Q_1\cup Q_2$ with
	\begin{align*}
	Q_1(x,y,v)&\leftarrow R_1(x,z),R_2(z,y),R_3(y,v),R_4(v,w)\text{ and }\\
	Q_2(x,y,v)&\leftarrow R_1(w,v),R_2(v,y),R_3(y,z),R_4(z,x).	
	\end{align*}
Since $Q_1$ and $Q_2$ are body-isomorphic,  $Q$ can be rewritten as
\begin{align*}
Q_1(w,y,z),Q_2(x,y,v)&\leftarrow R_1(w,v),R_2(v,y),R_3(y,z),R_4(z,x)
\end{align*}
In this case we can use the same approach used for single CQs in Theorem~\ref{theorem:originalDichotomy}, and show that this UCQ is not in $\DelayClin$ assuming \matmul. Let $A$ and $B$ be binary representations of
Boolean $n\times n$ matrices as explained in the beginning of this section. Define 
a database instance $I$ with $R_1^I = A$, $R_2^I=B$, $R_3^I=\{1,\ldots,n\}\times \{\bot\}$ and $R_4^I=\{(\bot,\bot)\}$.
$Q_1(I)$ corresponds to the answers of $AB$,
and $|Q_2(I)|=\calO(n^2)$.
Assume by contradiction that $\pEnum{Q}\in\DelayClin$. Then we find $Q(I)$ in $\calO(n^2)$ time.
We can distinguish the answers of $Q_1$ from those of $Q_2$ since the possible assignments to the variables in different queries are from disjoint domains: given an answer $\mu \in Q(I)$, we have that $\mu\in Q_1(I)$ if and only if $\mu(v)=\bot$.
This means we can solve matrix multiplication in time $\calO(n^2)$, which is a contradiction to \matmul.
\qed
\end{example}

A union of two difficult body-isomorphic acyclic CQs
may also be tractable. In fact, by adding a single variable to the heads in
Example~\ref{example:body-iso-matmul}, we obtain a tractable UCQ.

\begin{example}\label{example:hards-are-easy}
Let $Q$ be the UCQ
\begin{align*}
Q_1(w,y,x,z),Q_2(x,y,w,v)\leftarrow R_1(w,v),R_2(v,y),R_3(y,z),R_4(z,x)
\end{align*}
Both CQs are acyclic non-free-connex.
As $Q_2$ supplies the variables $\{v,w,y\}$ and $Q_1$
supplies $\{x,y,z\}$, both CQs have free-connex union extensions:
\begin{align*}
Q_1^+(w,y,x,z)\leftarrow R_1(w,v),R_2(v,y),R_3(y,z),R_4(z,x), P_1(v,w,y),\\
Q_2^+(x,y,w,v)\leftarrow R_1(w,v),R_2(v,y),R_3(y,z),R_4(z,x), P_2(x,y,z).	
\end{align*}
By Theorem~\ref{thm:positive} it follows that $\pEnum{Q}\in\DelayClin$.
\qed
\end{example}

Intuitively, the reason why the reduction of Example~\ref{example:body-iso-matmul} fails in Example~\ref{example:hards-are-easy} is the fact that all the variables of the free-paths in one CQ, which are used to encode matrix multiplication, are free in the other CQ. Indeed, if we encode matrices $A$ and $B$ to the relations of the free path $w,v,y$ in $Q_1$, there can be $n^3$ answers to $Q_2$. The answer set in this case is too large to contradict
the assumed lower bound for matrix multiplication. 
As it turns out, there are cases where we cannot reduce matrix multiplication to a union in this manner, and yet we can show that it is intractable using an alternative problem.

\begin{example}\label{example:acyclic-4clique}
	Let $Q$ be the UCQ
\begin{equation*}
	Q_1(x,y,u), Q_2(x,y,z)\leftarrow R_1(x,u,z),R_2(y,u,z).	
	\end{equation*}
	This union is intractable under the \fourclique~assumption.
	For a given graph $G=(V,E)$ with $|V|=n$,
	we compute the set $T$ of all triangles in $G$ in time $O(n^3)$.
	Define a database instance $I$ as
	$R_1^I=R_2^I=T$. For every output $\mu|_{\free(Q_i)}$ in $Q(I)$ with 
	$i\in\{1,2\}$, we know that $(\mu(x),\mu(u),\mu(z))$ and $(\mu(y),\mu(u),\mu(z))$ are triangles. If $\mu(x)\neq\mu(y)$, this means that
	$\mu((x,y))\in E$ if and only if $(\mu(x),\mu(y),\mu(u),\mu(z))$ forms a 4-clique (see Figure~\ref{figure:fourclique}).
	Moreover, for every 4-clique $(a,b,c,d)$ in the graph, we have $(a,b,c)\in Q_1(I)$, so this will detect all 4-cliques.
	Since there are $O(n^3)$ answers to $Q$, if $\pEnum{Q}\in\DelayClin$,
	we can check whether $\mu((x,y))\in E$ for every answer in $Q(I)$, and determine whether a $4$-clique appears in $G$ in time $\calO(n^3)$.
	\qed
\end{example}

\begin{figure}
\begin{tikzpicture}
[
    dot/.style={draw, circle, inner sep=1pt, fill= black},
    he/.style={draw, rounded corners,inner sep=0pt},        
    ce/.style={draw,dashed, rounded corners=2pt}, 
]

\node (x) at (0,0.5) {$x$};
\node (z) at (1,1) {$z$};
\node (w) at (1,0) {$w$};
\node (y) at (2,0.5) {$y$};

\draw [rounded corners = 3pt] 
			($(x.south west) + (-0.05,-0.05)$) -- ($(x.north west) + (-0.05,0.05)$) --
								 ($(z.north west) + (-0.05,0.05)$) -- ($(z.north east) + (0.05,0.05)$) --
								  ($(w.south east) + (0.05,-0.05)$) --  ($(w.south west) + (-0.05,-0.05)$) -- cycle;
\draw [rounded corners = 3pt] 
			($(y.north east) + (0.05,0.05)$) -- ($(y.south east) + (0.05,-0.05)$) --
			 ($(w.south east) + (0.00,-0.00)$) --  ($(w.south west) + (-0.00,-0.00)$) --
			 ($(z.north west) + (-0.00,0.00)$) -- ($(z.north east) + (0.00,0.00)$) -- cycle;
													 													 

\node[dot, label=$\mu(x)$] (x_1) at (4,0.5) {};
\node[dot, label=right:$\mu(w)$] (w_1) at (4.7,0) {};
\node[dot, label=$\mu(z)$] (z_1) at (4.7,1) {};	
\node[dot, label=right:$\mu(y)$] (y_1) at (5.4,0.5) {};

\draw (x_1) -- (z_1);
\draw 	  (x_1) -- (w_1);
\draw 	  (z_1) -- (y_1);
\draw 	  (z_1) -- (w_1);
\draw 	  (w_1) -- (y_1);


\node(H) at (-1, 0.5) {$\calH(Q):$};
\node (T) at (3,0.5) {$\mu(\vec{v}):$};
\end{tikzpicture}
\caption{
If $\mu|_{\free(Q_i)}\in Q(I)$, the induced subgraph of $\mu(w,x,y,z)$ forms a clique where one edge might be missing.}
\label{figure:fourclique}
\end{figure}

Note that we can use the \fourclique~assumption in Example~\ref{example:acyclic-4clique}, since, in addition to the free-path variables, there is another variable in both free-path relations.
We now generalize the observations from the examples.

\begin{definition}\label{def:two-guarded}
Let $Q=Q_1\cup Q_2$ be a UCQ where $Q_1$ and $Q_2$ are body-isomorphic.
\begin{itemize}
\item 	$Q_1$ is said to be \emph{free-path guarded} if for every free-path $P$ in $Q_1$, we have that $\var(P)\subseteq\free(Q_2)$.
\item Consider a path
$P=(u_1,\ldots u_k)$ in $Q_1$. Two atoms $R_1(\vec{v}_1)$
and $R_2(\vec{v}_2)$ of $Q_1$ are called \emph{subsequent $P$-atoms} if
$\{u_{i-1},u_{i}\}\subseteq\vec{v}_1$ and $\{u_i,u_{i+1}\}\subseteq\vec{v}_2$ for some $1<i<k$.
\item $Q_1$ is said to be \emph{bypass guarded} if for every free-path $P$ in $Q_1$ and variable $u$ that appears in two subsequent $P$-atoms, we have that $u\in\free(Q_2)$.
\end{itemize}
Note that every free-connex CQ is trivially free-path guarded and bypass guarded.
\end{definition}

Let us demonstrate the definitions on the recent examples.
The CQ $Q_1$ of Example~\ref{example:hards-are-easy} is both free-path guarded and bypass guarded.
The only free-path it contains is $P=(w,v,y)$, and $Q_1$ is free-path guarded since $\{w,v,y\}\subseteq\free(Q_2)$.
The only subsequent $P$-atoms are $R_1(w,v)$ and $R_2(v,y)$, and $Q_1$ is bypass guarded since $v\in\free(Q_2)$.
The query $Q_1$ of Example~\ref{example:body-iso-matmul} is not free-path guarded as the variables of the free-path $P'=(w,v,y)$
of $Q_1$ are not contained in $\free(Q_2)$.
The query $Q_1$ of Example~\ref{example:acyclic-4clique} is not bypass guarded.
Consider the free-path $P''=(x,z,y)$ of $Q_1$.
The atoms $R_1(x,u,z)$ and $R_2(y,u,z)$ are subsequent $P''$ atoms.
Since $u$ is contained in both atoms but not in $\free(Q_2)$, we have that $Q_1$ is not bypass guarded.

In the following two lemmas, we prove that if some CQ in a union is either
not free-path guarded or bypass guarded, then the UCQ is intractable.
As in the characterization for CQs in Theorem~\ref{theorem:originalDichotomy}, our lower bounds apply only when the CQs are self-join-free; with this restriction, we can assign different atoms with different relations.
The next lemma shows that the reduction in Example~\ref{example:body-iso-matmul},
where we can use the fact that $Q_1$ is not free-path guarded to
compute matrix multiplication, can be constructed in the general case as well.

\begin{lemma}\label{lemma:guardedTwoQueries}
	Let $Q=Q_1\cup Q_2$ be a UCQ of self-join-free body-isomorphic acyclic CQs. If $Q_1$ is not free-path guarded,
	then $\pEnum{Q}$ is not in $\DelayClin$, assuming \matmul.
\end{lemma}
\begin{proof}
Let $A$ and $B$ be Boolean $n\times n$ matrices represented as binary relations,
i.e. $A\subseteq\{1,\ldots,n\}^2$, where $(a,b)\in A$ means that the entry in the $a$th row
and $b$th column is $1$. Further let $P=(z_0,\ldots,z_{k+1})$ be a free-path in $Q_1$ that is
not guarded, and let $i$ be the minimal index such that $z_i\not\in\free(Q_2)$.
We assign the path variables to three roles as follows.
If $0<i<k+1$, we define $V_a=\{z_0,\ldots,z_{i-1}\}$, $V_b=\{z_i\}$ and $V_c=\{z_{i+1}\ldots,z_{k+1}\}$. Otherwise (if $i$ is $0$ or $k+1$), we define $V_a=\{z_0\}$, $V_b=\{z_1,\ldots,z_k\}$ and $V_c=\{z_{k+1}\}$.
Note that, in both cases, there exists some $\alpha\in\{a,b,c\}$ such that $V_\alpha\cap\free(Q_2)=\emptyset$. Intuitively that means that some role is not guarded.
Note also that since $P$ is chordless and $k\geq 1$, there is no atom in $Q$ that contains both a variable from $V_a$ and a variable from $V_c$.
Thus we can partition the atoms into nonempty sets
$\calR_A=\{R(\vec{v})\in\atoms(Q)\mid V_c\cap \vec{v}=\emptyset\}$ and $\calR_B=\atoms(Q)\setminus\calR_A$, and we have that the atoms of $\calR_A$ do not contain variables of $V_c$, and the atoms of $\calR_B$ do not contain variables of $V_a$.

Given three values $(a,b,c)$ we define a function $\tau_{(a,b,c)}:\var(Q)\rightarrow\{a,b,c,\bot\}$ that assigns every variable with the value corresponding to its role:
	\begin{equation*}
	\tau_{(a,b,c)}(v) = \left\{
     \begin{array}{ll}
       a &\text{ if }v\in V_a,\\
       b &\text{ if }v\in V_b,\\
       c &\text{ if }v\in V_c,\\
       \bot & \text{ otherwise,}
     \end{array}
   \right.	
	\end{equation*}
For a vector $\vec{v}$, we denote by $\tau_{(a,b,c)}(\vec{v})$ the vector obtained by element-wise application of $\tau_{(a,b,c)}$.
We define a database instance $I$ over $Q$ as follows: For every atom $R(\vec{v})$, if $R(\vec{v})\in \calR_A$ we set
$R^I = \{\tau_{(a,b,\bot)}(\vec{v})\mid (a,b)\in A\}$,
and if $R(\vec{v})\in \calR_B$ we set 
$R^I = \{\tau_{(\bot,b,c)}(\vec{v})\mid (b,c)\in B\}$. 
Note that every relation is defined only once since $\calR_A$ and $\calR_B$ are disjoint and $Q$ is self-join-free.

Consider an answer $\mu\in Q(I)$.
In the case that $0<i<k+1$, we have that $\mu(z_0)=\cdots=\mu(z_{i-1})=a$,
$\mu(z_i)=b$ and $\mu(z_{i+1})=\cdots=\mu(z_{k+1})=c$ for some $(a,b)\in A$ and $(b,c)\in B$,
and in case that $i\in\{0,k+1\}$, we have that $\mu(z_0)=a$,
$\mu(z_1)=\cdots=\mu(z_k)=b$ and $\mu(z_{k+1})=c$ for some $(a,b)\in A$ and $(b,c)\in B$.
This is since the variables $z_i$ are connected via the path in both CQs.
In either case, $\mu(\free(Q_1))$ is a tuple only containing the values $a,c$ and $\bot$. So the answers to $Q_1$ represent the answer to the matrix multiplication task we started with. We now need to verify that the answers to $Q_2$ do not interfere with the reduction.
If $0<i<k+1$, $\mu(\free(Q_2))$ is a tuple only containing the values $\{a,c,\bot\}$; if $i=0$, it only contains $\{b,c,\bot\}$; and if $i=k+1$, it only contains $\{a,b,\bot\}$.
Thus the number of answers to $Q$ is at most of size $2n^2$.
Assume by contradiction that we can enumerate the solutions of $Q(I)$ with linear preprocessing
and constant delay.
To distinguish the answers of $Q_1$ from those of $Q_2$, we can concatenate the variable names to the values, as described in Lemma~\ref{lemma:diff-domain}.
That way, we can ignore the answers that correspond to the values $(a,b)$ or $(b,c)$, and use the $(a,c)$ pairs as the answers to
matrix multiplication. This solves matrix multiplication in $O(n^2)$ time, in contradiction to \matmul.
\end{proof}

In Example~\ref{example:acyclic-4clique}, we encounter a UCQ where
both CQs are free-path guarded, but $Q_1$ is not bypass guarded. We can encode \fourclique~in every
UCQ with this property.

\begin{lemma}\label{lemma:interleaving-paths}
	Let $Q=Q_1\cup Q_2$ be a UCQ of self-join-free body-isomorphic acyclic CQs.
	If $Q_1$ and $Q_2$ are free-path guarded and $Q_1$ is not bypass guarded,
	then $\pEnum{Q}\not\in\DelayClin$, assuming \fourclique.
\end{lemma}
\begin{proof}
	Let $G=(V,E)$ be a graph with $|V|=n$.
	We show how to solve the \fourclique~problem
	on $G$ in time $\calO(n^3)$ if $\pEnum{Q}$
	is in $\DelayClin$.
	Let $P$ be a free-path in $Q_1$ and let $u\not\in\free(Q_2)$ such that 
	$u$ appears in two subsequent P-atoms.
	
	We first claim that, under the conditions of this lemma,
	$P$ is of length $2$.	
	Let $P=(z_0,\ldots,z_{k+1})$ and $1\leq i \leq k$ such that
    $\{u,z_{i-1},z_{i}\}$ and $\{u,z_{i},z_{i+1}\}$ are contained in edges of $\calH(Q)$. 
	As $P$ is chordless, there is no edge containing $\{z_{i-1},z_{i+1}\}$ , thus
	the path $(z_{i-1},u,z_{i+1})$ is a chordless path. As $Q_1$ is free-path guarded,
	 $z_{i-1},z_{i+1}\in\free(Q_2)$ and since $u\not\in\free(Q_2)$,
	this is a free-path of $Q_2$. Since $Q_2$ is free-path guarded we have that $z_{i-1},z_{i+1}\in\free(Q_1)$.
	Since $P$ is a free-path in $Q_1$, the only variables of $P$ that are free in $Q_1$ appear in the edges of the path, and so $i=k=1$, and $P$ is of the form $(z_0,z_1,z_2)$.	
Therefore, there exist atoms $R_1$ and $R_2$ with $\{z_0,z_1,u\}\subseteq\var(R_1)$ and $\{z_1,z_2,u\}\subseteq\var(R_2)$.

Given three values $(a,b,c)$ we define a function $\tau_{(a,b,c)}:\var(Q)\rightarrow\{a,b,c\}$ as follows:
	\begin{equation*}
	\tau_{a,b,c}(v) = \left\{
     \begin{array}{ll}
       a &\text{ if }v\in\{z_0,z_2\},\\
       b &\text{ if }v=z_1,\\
       c &\text{ if }v=u,\\
       \bot & \text{ otherwise.}
     \end{array}
   \right.	
	\end{equation*}
For every atom $R(\vec{v})\in\atoms(Q)$, we define
$R^I=\{\tau_{a,b,c}(\vec{v})\mid (a,b,c)\text{ a triangle in }G\}$.
Every relation is defined only once since $Q$ is self-join-free.
Note that $|R^I|\in\calO(n^3)$, as there are at most $n^3$ triangles in $G$, and that we can construct $I$ with $\calO(n^3)$ time.

Consider an answer $\mu|_{\free(Q_1)}\in Q_1(I)$. Since $R_1$ and $R_2$ are atoms in $Q_1$, we are guaranteed that $(\mu(z_0),\mu(z_1),\mu(u))$
and $(\mu(z_1),\mu(z_2),\mu(u))$ form triangles in $G$. Therefore, the graph contains a 
4-clique if and only if there is an edge $(\mu(z_0),\mu(z_2))$.
As $z_0,z_2\in\free(Q_1)$ it suffices to check every $\mu|_{\free(Q_1)}\in Q_1(I)$
for this property.
Since $Q_1$ is free-path guarded, we know that $z_1$ is existential in $Q_1$ but free in $Q_2$. This means that $\free(Q_1)\neq\free(Q_2)$ and so there is no homomorphism from $Q_2$ to $Q_1$.
We can therefore use the mappings from Lemma~\ref{lemma:diff-domain} in order to distinguish the answers of $Q_1$ from those of $Q_2$.
We have that $\{z_0,z_1,z_2,u\}$ is neither contained in $\free(Q_1)$
nor in $\free(Q_2)$. Thus, $|Q(I)|\in\calO(n^3)$.
If $\pEnum{Q}$ is in $\DelayClin$, we can construct the database instance, compute $Q$,
 and check every answer to $Q_1$ for an edge of the form $(\mu(z_0),\mu(z_2))$ in total time $\calO(n^3)$, which contradicts \fourclique.
\end{proof}

In the next section, we show that Lemma~\ref{lemma:guardedTwoQueries} and Lemma~\ref{lemma:interleaving-paths} cover all intractable cases of the UCQs discussed in this section.

\subsection{Complete Classification of Fragments of UCQs}

We show next that the hardness proofs given above complete our tractability results into a dichotomy for the fragments of UCQs we examined. We first establish this for unions of two body-isomorphic acyclic CQs, and then conclude the same for unions of two intractable CQs.

To prove that any union of two body-isomorphic acyclic CQs that is not covered by Lemma~\ref{lemma:guardedTwoQueries} and Lemma~\ref{lemma:interleaving-paths} has a free-connex union extension, we need some observations regarding the place of appearance of relevant variables in join-trees. Recall that we call $(v_1,\ldots,v_n)$ a path of variables in a query $Q$ if
for all $1 \leq i < n$, there exists an atom $R(\vec{u_i})$ in $Q$ such that $\{v_i,v_{i+1}\}\subseteq \vec{v}$.

\begin{proposition}\label{prop:path-vars}
Let $(v_1,\ldots,v_n)$ be a path of variables in an acyclic query $Q$,
and let $A_s$ and $A_t$ be atoms containing $\{v_1,v_2\}$ and $\{v_{n-1},v_n\}$ respectively.
For all $1 \leq i < n$, the simple path between $A_s$ and $A_t$  on a join-tree of $Q$ contains an atom $R_i(\vec{u_i})$ such that $\{v_i,v_{i+1}\}\subseteq\vec{u_i}$.
\end{proposition}
\begin{proof}
We prove this by induction on $n$.
If $n=3$, this trivially holds as the endpoints, $A_s$ and $A_t$, contain the required variables.
We now assume this proposition holds for paths of length $n-1$ and show that it also holds for paths of length $n$ where $n\geq 4$.
Consider the simple path between $A_s$ and a node containing $\{v_{n-2},v_{n-1}\}$. Let $A_m$ be the first node on that path that contains $\{v_{n-2},v_{n-1}\}$.
Denote by $P_s$ the simple path between $A_s$ and $A_m$. Note that $A_m$ is the only node on $P_s$ that contains $\{v_{n-2},v_{n-1}\}$.
By the induction assumption, $P_s$ contains $\{v_i,v_{i+1}\}$ for all $1 \leq i < n-1$, and in particular it contains $\{v_{n-3},v_{n-2}\}$.
Denote by $P_t$ the simple path between $A_m$ and $A_t$, and denote by $P$ the concatenation of $P_s$ and $P_t$. The path $P$ goes between $A_s$ and $A_t$ and contains $\{v_i,v_{i+1}\}$ for all $1 \leq i < n$. We claim that $P$ is a simple path.
Assume by contradiction that $P$ is not simple. This means that there is a node $A_v$ that $P_s$ and $P_t$ share other than $A_m$. Due to the running intersection property, since $A_v$ is on the simple path between $A_m$ and $A_t$, $A_v$ contains $v_{n-1}$.
Denote by $A_{m-1}$ the node preceding $A_m$ on $P_s$. Due to the running intersection property, since $A_{m-1}$ is on the simple path between $A_v$ and $A_{m}$, $A_{m-1}$ contains $v_{n-1}$ too.
Since a node containing $\{v_{n-3},v_{n-2}\}$ is on $P_s$, and since $A_{m-1}$ is on the simple path between this node and $A_m$, we have that $A_{m-1}$ contains $v_{n-2}$.
Overall, we have that $A_{m-1}$ contains $\{v_{n-2},v_{n-1}\}$, which is a contradiction to the fact that $A_m$ is the only such node on $P_s$.
\end{proof}

\begin{proposition}\label{prop:var-on-path}
Let $(v_1,\ldots,v_n)$ be a path of variables in an acyclic query $Q$,
and let $A_s$ and $A_t$ be atoms containing $\{v_1,v_2\}$ and $\{v_{n-1},v_n\}$ respectively.
If $A_1$ and $A_2$ are two subsequent nodes on the simple path $P$ between $A_s$ and $A_t$  on a join-tree of $Q$, then there exists some $1 \leq i \leq n$ such that $v_i\in A_1\cap A_2$.
\end{proposition}
\begin{proof}
Denote the subpath of $P$ between $A_s$ and $A_1$ by $P_s$, and the subpath between $A_2$ and $A_t$ by $P_t$. 
Denote by $i$ the maximal index such that $v_i \in\var(P_s)$.
If $i=n$, note that an atom in $P_s$ and an atom in $P_t$ both contain $v_i$.
Otherwise, $i<n$. According to Proposition~\ref{prop:path-vars} and since all atoms of $P$ appear in either $P_s$ or $P_t$, an atom containing the variables $\{v_i,v_{i+1}\}$ must appear in $P_s$ or $P_t$.
It does not appear in $P_s$ because of the maximality of $i$, so it appears in $P_t$. In this case too, an atom in $P_s$ and an atom in $P_t$ both contain $v_i$.
Since $A_1$ and $A_2$ are on the path between these atoms, due to the running intersection property, $v_i\in A_1\cap A_2$.

\end{proof}

Using this proposition, we can prove the structural property that we need.

\begin{lemma}\label{lemma:fully-contracted}
Let $Q=Q_1\cup Q_2$ be a UCQ of body-isomorphic acyclic CQs, where $Q_1$ and $Q_2$ are free-path guarded, and $Q_1$ is bypass guarded,
and let $P=(z_0,\ldots,z_{k+1})$ be a free-path in $Q_1$.
There exists a join-tree $T$ for $Q$ with a subtree $T_P$ such that:
\begin{itemize}
\item $\var(P)\subseteq\var(T_P)$.
\item For every variable $u$ that appears in two different atoms of $T_P$:
\begin{itemize}
\item $u\in\free(Q_2)$.
\item There is an atom $R(\vec{v})$ in $T_P$ such that $u,z_i \in\vec{v}$ for some $0<i<k+1$.
 \end{itemize}
\end{itemize}
\end{lemma}

\begin{proof}
Recall that according to our notation, we assume that body-isomorphic CQs have exactly the same body. Thus, $T$ is also a join-tree for $Q_2$. In the following, we refer to the body of $Q$ when we make statements that apply to the bodies of both $Q_1$ and $Q_2$.

Consider a path $P_A=(A_1,\ldots,A_s)$ between two atoms on a join tree.
We define a \emph{contraction step} for a path of length $2$ or more:
 if there exists $j$ such that $A_j\cap A_{j+1} \subseteq A_1\cap A_s$, then remove the edge $(A_j,A_{j+1})$ and add the edge $(A_1,A_s)$.
A path on a join-tree is said to be \emph{fully-contracted} if none of its subpaths can be contracted.
Given any two atoms on a join-tree, it is possible to fully contract the path between them by performing any arbitrary sequence of contraction steps until it is no longer possible: the process will end as every contraction step reduces the length of the path.

We now claim that the graph $T'$ obtained from such a contraction step of a path $P_A$ on a join-tree $T$, remains a join-tree.
It is still a tree since it remains connected and with the same number of edges as before.
It is left to show the running intersection property.
We start with some observations regarding $T$. Since the running intersection property holds in $T$, for all $1\leq i \leq s$, $A_1\cap A_s \subseteq A_i$. Since $A_j\cap A_{j+1} \subseteq A_1\cap A_s$, we also have that $A_j\cap A_{j+1} \subseteq A_i$.
Now consider two nodes $B$ and $C$. We need to show that every node on the simple path between them in $T'$ contains $B\cap C$. If it is the same path as in $T$, then we are done. Otherwise, a path between $B$ and $C$ in $T'$ can be obtained by using the simple path between them in $T$ and replacing the edge $(A_j,A_{j+1})$ with $(A_{j+1},\ldots,A_s,A_1,\ldots,A_{j})$. The simple path between $B$ and $C$ is contained in this path. 
This means that atoms on the simple path between $B$ and $C$ in $T'$ are either: (1) on the simple path between $B$ and $C$ in $T$, and therefore contain $B\cap C$; (2) on the path $P_A$ and therefore contain $A_j\cap A_{j+1}$. Since $A_j$ and $A_{j+1}$ are on the path between $B$ and $C$ in $T$, we have that $B\cap C \subseteq A_j\cap A_{j+1}$, so in this case too, the atoms contain $B\cap C$.
This proves that the contracted graph is indeed a join-tree. By induction, if we fully contract a path on a join-tree, we still have a join-tree.

Now let $T$ be a join-tree of $Q$. We consider some path in $T$ between an atom containing $\{z_0,z_1\}$ and an atom containing $\{z_k,z_{k+1}\}$. Take the unique subpath of it containing only one atom with $\{z_0,z_1\}$ and one atom with $\{z_k,z_{k+1}\}$, and fully contract it. We denote this fully contracted subpath as $T_P=(A_1,\ldots,A_s)$.
Due to Proposition~\ref{prop:path-vars}, $\var(P)\subseteq\var(T_P)$.

First, we claim that every variable $u$ that appears in two or more atoms of $T_P$ is part of a chordless path from $z_0$ to $z_{k+1}$.
We first show a chordless path from $u$ to $z_{k+1}$. Denote the last atom on $T_P$ containing $u$ by $A_i$. 
If $A_i$ contains $z_{k+1}$, we have found the chordless path $(u,z_{k+1})$, and we are done. Otherwise, $A_i$ is not the last atom on $T_P$. It is also not the first atom on $T_P$, as another atom contains $u$. Consider the subpath $A_{i-1},A_i,A_{i+1}$. Since it is fully contracted, $A_i\cap A_{i+1} \not\subseteq A_{i-1}\cap A_{i+1}$. This means that there is a variable $u_{+1}$ in $A_i$ and in $A_{i+1}$ that does not appear in $A_{i-1}$.
Now consider the last atom containing $u_{+1}$, and continue with the same process iteratively until reaching $z_{k+1}$.
This results in a chordless path $u,u_{+1},\ldots,u_{+m}=z_{k+1}$ with $m\geq 1$.
Do the same symmetrically to find a chordless path $z_0=u_{-n},\ldots, u_{-1},u$ with $n\geq 1$.
We now claim that the concatenation $P_u=(u_{-n},\ldots, u_{-1},u,u_{+1},\ldots,u_{+m})$ is chordless.
We prove that $u_{-t}$ and $u_{+\ell}$ are not neighbors for all $1\leq t\leq n$ and $1\leq \ell\leq m$ by induction on $\ell+t$.
Out of the atoms containing $u$, the variable $u_{+1}$ only appears in the last atom by construction: if $u_{+1}=z_{k+1}$ this is true since $z_{k+1}$ appears only in one atom, and otherwise it is true because this is how we chose $u_{+1}$. Similarly $u_{-t}$ only appears in the first. Since there are at least two atoms of the path containing $u$, we have that $u_{-1}$ and $u_{+1}$ are not neighbors, and this proves the induction base.
Next, assume by way of contradiction that $u_{-t}$ and $u_{+\ell}$ are neighbors.
By using the induction assumption, we have that $u_{-t},\ldots,u_{-1},u,u_{+1},\ldots,u_{+\ell},u_{-t}$ is a chordless cycle of length four or more, contradicting the fact that $Q$ is acyclic and therefore chordal.
This proves that $u$ is part of the chordless path $P_u$ from $z_0$ to $z_{k+1}$.

Assume by contradiction that a variable $u\not\in\free(Q_2)$ appears in two distinct atoms of $T_P$. There is a chordless path from $z_0$ to $z_{k+1}$ that contains $u$.
Denote this path $P_u$. Take a subpath of $P_u$ starting with the last variable on $P_u$ before $u$ that is in $\free(Q_2)$, and ending with the first variable on $P_u$ after $u$ that is in $\free(Q_2)$.
Such variables exist on this path because $z_0,z_{k+1}\in\free(Q_2)$. This subpath is a free-path in $Q_2$, and since $Q_2$ is free-path guarded, $u\in\free(Q_1)$.
Next consider two neighboring atoms on $T_P$ that contain $u$.
According to Proposition~\ref{prop:var-on-path}, there exists some $z_i$ that appears in both atoms. Note that $i>0$ and $i<k+1$ since the path only contains one atom with $z_0$ and one atom with $z_{k+1}$.
Since $Q_1$ is bypass guarded and $u\not\in\free(Q_2)$, it is not possible that there is both an atom with $\{z_{i-1},z_i,u\}$ and an atom with $\{z_i,z_{i+1},u\}$ in $Q$.
Without loss of generality, assume there is no atom with $\{z_i,z_{i+1},u\}$. Since the query is acyclic, this means that $u$ and $z_{i+1}$ are not neighbors (if the three variables $z_i,z_{i+1},u$ are pairwise neighbors in an acyclic graph, then necessarily they all appear in the same hyperedge). 
Then, there is a chordless path $(u,z_i,z_{i+1},\ldots,z_{k+1})$. Since $u\in\free(Q_1)$, it is a free-path. This contradicts the fact that $Q_1$ is free-path guarded since $u\not\in\free(Q_2)$.

It is left to show that there is an atom $R(\vec{v})$ in $T_P$ such that $u,z_i \in\vec{v}$ for some $0<i<k+1$. Since $u$ appears in two distinct atoms of $T_P$, and since $T_P$ is a join-tree, $u$ also appears in two adjacent atoms of $T_P$. According to Proposition~\ref{prop:var-on-path}, there exists $0\leq i \leq k+1$ such that $z_i$ is in those atoms.
This cannot be $0$ or $k+1$ because they appear only in one atom in $T_P$.
\end{proof}

We are now ready to show that the properties free-path guardedness and
bypass guardedness imply free-connexity.

\begin{lemma}\label{lemma:dichotomyTwoQueriesPositive}
	Let $Q=Q_1\cup Q_2$ be a UCQ of body-isomorphic acyclic CQs. If $Q_1$ and $Q_2$ are both free-path guarded and bypass guarded, then $Q$ is free-connex.
\end{lemma}

\begin{proof}
We describe how to iteratively build a union extension for each CQ.
In every step, we take one free-path among the queries in $Q$ and add a virtual atom
in order to eliminate this free-path.
More specifically, let $P=(z_0,\ldots,z_{k+1})$ be a free-path in $Q_1$.
Take $T_P$ according to Lemma~\ref{lemma:fully-contracted}, and denote by $V_P$ the variables of $P$ and all variables that appear in more than one atom of $T_P$. 
We add the atom $R(V_P)$ to both $Q_1$ and $Q_2$ and obtain $Q_1^+$ and $Q_2^+$ respectively.
We will show that repeatedly applying such steps eventually leads to free-connex union extensions for $Q_1$ and $Q_2$.

First, we claim that $Q_2$ supplies $V_P$. It is guaranteed that $V_P\subseteq\free(Q_2)$ since $P$ is guarded and due to Lemma~\ref{lemma:fully-contracted}. We now show that $Q_2$ is acyclic $V_P$-connex.
We know that $T_P$ is a subtree of $T$ that contains $V_P$, but it may contain additional variables, each appearing in only atom, so we need to modify the tree.
For every vertex $A_i$ in $T_P$, add another vertex with $A_i'=\var(A_i)\cap V_P$ and an edge $(A_i,A_i')$.
Then, for every edge $(A_i,A_j)$ in $T_P$, replace it with the edge $(A_i',A_j')$. The new graph is a tree since it is connected and the number of added vertices is equal to the number of added edges (hence, after the modification, the number of edges remains equal to the number of vertices minus one).
We next show that the running intersection property is maintained. For every edge $(A_i,A_j)$ removed, $\var(A_i)\cap\var(A_j)\subseteq V_P$, and so by definition of the new vertices, $\var(A_i)\cap\var(A_j)=\var(A_i')\cap\var(A_j')$.
Given two nodes $B$ and $C$ of the join-tree, the path between them in the modified join-tree is similar to the path in the original join-tree, except it may contain the node $(A_i',A_j')$ if it contained the node $(A_i,A_j)$ before.
Since the running intersection property holds in $T$, 
every node on the path between $B$ and $C$ contains $\var(B)\cap \var(C)$.
Since $\var(A_i)\cap\var(A_j)=\var(A_i')\cap\var(A_j')$,
the running intersection property also holds in the modified tree.
As the variables of the subtree that consists of the new vertices are exactly $V_P$, this concludes the proof that $Q_2$ supplies $V_P$.

After the extension step we described, there are no free-paths that start in $z_0$ and end in $z_{k+1}$ since these variables are now neighbors. If both of the CQs are now free-connex, then we are done.
Otherwise, we use the extension recursively. We can apply the extension again as we show next that the UCQ $Q_1^+\cup Q_2^+$ conforms to the conditions of this lemma.
Note that after a free-path from $z_0$ to $z_{k+1}$ is treated, and even after future extension, there will be no free-path from $z_0$ to $z_{k+1}$ since they are now neighbors.
Since there is a finite number of variable pairs, after a finite number of such steps, all pairs that have a free-path between them are resolved. At this point, we have an acyclic extension with no free-paths, so it is free-connex.
It is left to prove that the conditions of this lemma are maintained after the extension steps as long as at least one of the extended CQs is not free-connex. We show that in the following three claims.

\begin{claim}\label{claim:lemmaClaimAcyclic}
$Q^+_1$ and $Q^+_2$ are body-isomorphic acyclic.
\end{claim}
\begin{proof}[Proof of Claim~\ref{claim:lemmaClaimAcyclic}]
First recall that we assume that the CQs are given in the notation defined in the beginning of Section~\ref{sec:2-body-iso}, so the two CQs have exactly the same body. Since we add the same atom to both CQs, the extensions also have the same body.
We show a join-tree $T^+$ for the extension. Take the join-tree $T$ according to Lemma~\ref{lemma:fully-contracted}, and add the vertex $R(V_P)$. Remove all edges in $T_P$ and add an edge between $R(V_P)$ and every atom in $T_P$. This construction results in a connected graph with no cycles, and so it is a tree. We claim that the running intersection property is preserved.
Let $B$ and $C$ be two nodes in $T^+$.
We first handle the case that none of $B$ and $C$ are $R(V_P)$.
We need to show that every node on the path between them in $T^+$ contains $B\cap C$. Since this property holds in $T$, every node on the path between them in $T^+$ that also appears on the path between them in $T$ preserves this property. By construction, the only new node that may appear on this path is $R(V_P)$. In this case, the two nodes before and after $R(V_P)$ on this path contain $B\cap C$.
By definition, $V_P$ contains all variables that appear in more that one atom in $T_P$, so $R(V_P)$ contains every variable that appears in more than one of its neighbors, and it also contains $B\cap C$.
It is left to handle the second case. Assume without loss of generality that $C=R(V_P)$. We need to handle the path between $R(V_P)$ and $B$.
Let $v\in V_P\cap B$. Since $v\in V_P$, there exists some node $A_v$ in $T_P$ that contains it.
Consider the simple path in $T$ between $A_v$ and $B$, and let $A_v'$ be the last node on this path which is in $T_P$.
Due to the running intersection property, every node on this path contains $A_v\cap B$, and so it also contains $v$. The edge from $V_P$ to $A_v'$ and the simple path from $A_v'$ to $B$ is therefore a simple path in $T^+$ from $V_P$ to $B$ that contains $v$.
\end{proof}

\begin{claim}\label{claim:lemmaClaimHPCovered}
$Q^+_1$ and $Q^+_2$ are free-path guarded.
\end{claim}
\begin{proof}[Proof of Claim~\ref{claim:lemmaClaimHPCovered}]
Since the original CQs are free-path guarded, every free-path in the extension that is also a free-path
in the original query is guarded.
According to our construction,
the only atom that was added in the extension contains exactly $V_P$. Thus, a new free-path
$(v_0,\ldots,v_{m+1})$ contains $v_j,v_{j+1}\in V_P\subseteq\free(Q_2)=\free(Q_2^+)$.
In particular, since the variables in the center of a free-path must be existential, $Q_2^+$ does not contain new free-paths. It is left to handle free-paths that appear in $Q_1^+$ but not in $Q_1$.

Let $P'=(v_0,\ldots,v_{m+1})$ be a free-path in $Q^+_1$ but not in $Q_1$, and let $v_j,v_{j+1}\in V_P$. We need to show that $v_i\in\free(Q_2^+)$ for all $0\leq i \leq m+1$. First note that $v_j,v_{j+1}\in V_P \subseteq\free(Q_2)=\free(Q_2^+)$. We next prove the same for $v_0$ and $v_{m+1}$.

We first claim that there is a path $P_{\text{mid}}$ between $v_j$ and $v_{j+1}$ that goes only through existential variables in $Q_1$.
We prove that every variable in $V_P$ is either of the form $z_i$ such that $0< i<k+1$ (recall that these are the variables in the center of the eliminated free-path) or it has a neighbor of that form.
Let $v\in V_P$. By definition of $V_P$, either there exists $0\leq i\leq k+1$ such that $v=z_i$ or $v$ appear in two atoms or more in $T_P$.
In the first case, if $0<i<k+1$, then it is of the required form. Otherwise, if $i=0$, it has the neighbor $z_1$, and if $i=k+1$, it has the neighbor $z_k$.
In the second case (if $v$ appear in two atoms or more in $T_P$), according to Lemma~\ref{lemma:fully-contracted}, $v$ has a neighbor of the required form.
Since $v_j,v_{j+1}\in V_P$, this proves that each of $v_j$ and $v_{j+1}$ is either of the form $z_i$ such that $0< i<k+1$ or it has a neighbor of that form.
Since all $z_i$ such that $0< i<k+1$ are connected through $P$, there is a path $P_{\text{mid}}$ between $v_j$ and $v_{j+1}$ that goes only through existential variables in $Q_1$.
 
Since $P'$ is chordless, no variable in $P'$ other than $v_j$ and $v_{j+1}$ is in $V_P$. One consequence of this is that $P_s=(v_0,\ldots,v_j)$ and $P_t=(v_{j+1},\ldots,v_{m+1})$ are also paths in $Q_1$.
Another consequence is that $v_0$ and $v_{m+1}$ are not both in $V_P$.
Since $P'$ is chordless and $m\geq 1$, we have that $v_0$ and $v_{m+1}$ are not neighbors in $Q^+$. Since they are not both on the new atom, $v_0$ and $v_{m+1}$ are not neighbors in $Q$ either.
By concatenating $P_s$, $P_{\text{mid}}$ and $P_t$, we obtain a path in $Q$ from $v_0$ to $v_{m+1}$ that goes only through existential variables in $Q_1$. Take a simple chordless path $P_\ell$ contained in it.
Since $v_0$ and $v_{m+1}$ are not neighbors, the path $P_\ell$ is of length two at least, and so it is a free-path in $Q_1$. Since $Q_1$ is free-path guarded, we conclude that $v_0,v_{m+1}\in\free(Q_2)$.

So far we have that $v_0,v_j,v_{j+1},v_{m+1}\in\free(Q_2)$.
Assume by contradiction that there exists $0<i<j$ such that $v_i\not\in\free(Q_2)$, and consider the subpath of $P_s$ beginning with the last variable before $v_i$ on $P_s$ that is in $\free(Q_2)$ and ending with the first after $v_i$ on $P_s$ that is in $\free(Q_2)$. This is a free-path in $Q_2$ containing $v_i$. We know that $v_i\not\in \free(Q_1)$ as it is in the center of the free-path $P'$, but this contradicts the fact that $Q_2$ is free-path guarded. In the same way, we can show that $v_i\in\free(Q_2)$ for all $j+1<i<m+1$.
Overall we have seen that $v_i\in\free(Q_2)=\free(Q_2^+)$ for all $0\leq i \leq m+1$ as required.
\end{proof}

\begin{claim}\label{claim:lemmaClaimHPNoClique}
$Q_1^+$ and $Q_2^+$ are bypass guarded.
\end{claim}
\begin{proof}[Proof of Claim~\ref{claim:lemmaClaimHPNoClique}]
First let $P'=(t_0,\ldots,t_{m+1})$ be a free-path in $Q^+_2$, and assume by contradiction that there exists some
$u\not\in\free(Q_1^+)$ that appears in two subsequent $P'$-atoms. This means that there exists $i$ such that $Q^+_2$ has an atom containing $\{t_{i-1},t_i,u\}$ and an atom containing $\{t_i,t_{i+1},u\}$.
As explained in the previous claim, $Q^+_2$ has no new free-paths, so $P'$ is a free-path in $Q_2$ as well. Since $Q_2$ is bypass guarded, $u$ does not appear in two subsequent $P'$-atoms in $Q_2$, so one of these atoms is new in $Q^+_2$. Assume without loss of generality that it is $\{t_i,t_{i+1},u\}$.
Then, $\{t_i,t_{i+1},u\}\subseteq V_P\subseteq\free(Q_2)$. This contradicts the fact that $t_i\not\in\free(Q_2)$ since $P'$ is a free-path.

Now let $P'=(t_0,\ldots,t_{m+1})$ be a free-path in $Q^+_1$. Assume by contradiction that there exists some $u\not\in\free(Q_2^+)$,
that appears in two subsequent $P'$-atoms. This means that $Q^+$ has an atom containing $\{t_{i-1},t_i,u\}$ and an atom containing $\{t_i,t_{i+1},u\}$. 
Since $P'$ is chordless, $t_{i-1}$ and $t_{i+1}$ are not neighbors, and so $(t_{i-1},u,t_{i+1})$ is a chordless path as well.
According to Claim~\ref{claim:lemmaClaimHPCovered}, $\var(P')\subseteq\free(Q_2^+)$, and so $(t_{i-1},u,t_{i+1})$ is a free-path in $Q_2^+$. According to Claim~\ref{claim:lemmaClaimHPCovered} again, $\{t_{i-1},u,t_{i+1}\}\subseteq\free(Q_1^+)$. Since the only free variables on a free-path are at its ends, this means that $P'$ is of length two, and $i=m=1$.
Since $u\not\in\free(Q_2)$ and ${V_P}\subseteq\free(Q_2)$, we have that the atoms containing $\{t_0,t_1,u\}$ and $\{t_1,t_2,u\}$ are not the added atom, and so they appear also in $Q_2$ and in $Q_1$.
This means that $P'$ is a free-path in $Q_1$, and $u\not\in\free(Q_2)$ appears in two subsequent atoms of $P'$ in $Q_1$.
Thus, $Q_1$ is not bypass guarded, in contradiction to the conditions of this lemma.

\end{proof}

This proves that the lemma can be applied iteratively.
\end{proof}

Since Lemma~\ref{lemma:guardedTwoQueries}, Lemma~\ref{lemma:interleaving-paths} and Lemma~\ref{lemma:dichotomyTwoQueriesPositive} cover all cases of a union of two self-join-free body-isomorphic acyclic CQs, we have a dichotomy that characterizes the fragment of UCQs we discuss.

\begin{theorem}\label{theorem:ufc_characterize}
	Let $Q=Q_1\cup Q_2$ be a UCQ of self-join-free body-isomorphic CQs.
	\begin{itemize}
	\item If $Q_1$ and $Q_2$ are acyclic free-path guarded and bypass guarded, 
	then $Q$ is free-connex and $\pEnum{Q}\in\DelayClin$.
	\item Otherwise, $Q$ is not free-connex and $\pEnum{Q}\not\in\DelayClin$, under the assumptions \hyperclique, \matmul~and \fourclique.
	\end{itemize}
\end{theorem}
\begin{proof}
Since the CQs are body-isomorphic, they are either both acyclic or both cyclic.
First assume that the CQs are acyclic.
By Lemma~\ref{lemma:dichotomyTwoQueriesPositive}, if $Q_1$ and $Q_2$ are both free-path guarded and bypass guarded, then $Q$ is free-connex, and $\pEnum{Q}\in\DelayClin$ by Theorem~\ref{thm:positive}.
Otherwise, either one of $Q_1$ and $Q_2$ is not free-path guarded, or they both are but one of them is not bypass guarded. In these cases, 
by Lemma~\ref{lemma:guardedTwoQueries} and Lemma~\ref{lemma:interleaving-paths}, $\pEnum{Q}\not\in\DelayClin$ assuming \matmul~and \fourclique.
By Theorem~\ref{theorem:union_of_intractable}, if the CQs are cyclic, $\pEnum{Q}$ is not in $\DelayClin$ assuming $\hyperclique$.
In all cases where $\pEnum{Q}\not\in\DelayClin$, we know that $Q$ is not free-connex by Theorem~\ref{thm:positive}.
\end{proof}

By combining Theorem~\ref{theorem:ufc_characterize} with Theorem~\ref{theorem:union_of_intractable}, we have the following dichotomy for the case of unions containing exactly two intractable CQs.

\begin{theorem}\label{thm:two-hard-dichotomy}
	Let $Q=Q_1\cup Q_2$ be a union of difficult CQs.
\begin{itemize}
	\item If $Q$ is free-connex, then $\pEnum{Q}\in\DelayClin$.
	\item	 If $Q$ is not free-connex, then $\pEnum{Q}\not\in\DelayClin$, assuming \matmul, \hyperclique~and \fourclique.
\end{itemize}	
\end{theorem}

\begin{proof}
If $Q$ is free-connex, Theorem~\ref{thm:easyunions} proves that it is tractable.
If $Q_1$ and $Q_2$ are acyclic and body-isomorphic but $Q$ is not free-connex, then Theorem~\ref{theorem:ufc_characterize} proves that $Q$ is intractable.
Otherwise, $Q$ is a union of difficult CQs that does not contain body-isomorphic acyclic CQs, and Theorem~\ref{theorem:union_of_intractable} proves that it is intractable.
\end{proof}

\section{Towards a Dichotomy}\label{sec:dichotomy}

In this section we examine the next steps that are required to fully characterize which UCQs are in $\DelayClin$. We pinpoint some of the difficulties that must be tackled when formulating such a dichotomy, accompanied by examples.
In Section~\ref{sec:future-acyclic} we discuss unions containing only acyclic CQs, and in Section~\ref{sec:future-cyclic} we discuss those that contain at least one cyclic CQ.

\subsection{Unions of Acyclic CQs}\label{sec:future-acyclic}

We inspect two ways of extending the results of Section~\ref{sec:2-body-iso}.
The first such extension is to a union of two CQs that are not body-isomorphic.
If there is a difficult CQ $Q_1$ in the union where for every other CQ $Q_i$ in the union there is no body-homomorphism from $Q_i$ to $Q_1$, we can reduce $Q_1$ to the union as described in Lemma~\ref{lemma:body-homo}. Since $Q_1$ is difficult, the union is intractable too. In case there is a body-homomorphism to the hard queries, one might think that it is sufficient for intractability to have unguarded difficult structures similarly to the previous section.
This is incorrect.

\begin{example}[Unresolved UCQ with CQs that are not body-isomorphic]\label{example:separated}
Let $Q=Q_1\cup Q_2$ with
\begin{align*}
	Q_1(x,y,w)&\leftarrow R_1(x,z),R_2(z,y),R_3(y,w)\text{ and}\\
	Q_2(x,y,w)&\leftarrow R_1(x,t_1),R_2(t_2,y),R_3(w,t_3).
\end{align*}
The query $Q_1$ is acyclic non-free-connex, while $Q_2$ is free-connex. Further $Q_1$ is not contained in $Q_2$ but there is a body-homomorphism from $Q_2$ to $Q_1$. The variable $z$ is part of the free-path $(x,z,y)$ in $Q_1$, but the variables $t_1$ and $t_2$ that map to it via the body-homomorphism are not free in $Q_2$. If we extend the notion of guarding to non-body-isomorphic CQs in the natural way, the free-path $(x,z,y)$ is not guarded.
Nevertheless, we cannot compute matrix multiplication in $O(n^2)$ time by encoding it to the free-path $(x,z,y)$ like before. Over such a construction, there can be $n^3$ many results to $Q_2$ as $w$ and $y$ are not connected in $Q_2$ so they can have distinct values.
This is not an issue when discussing body-isomorphic CQs.
We do not know whether this example is in $\DelayClin$.
\qed
\end{example}
A future characterization of the union of CQs that are not body-isomorphic would need an even more careful approach than the one used in Section~\ref{sec:2-body-iso} in order to handle the case that variables mapping to the free-path are not connected via other variables mapping to the free-path.

A second way of extending Section~\ref{sec:2-body-iso} is to consider more than two body-isomorphic acyclic CQs.
This case may be tractable or intractable, and for some queries of this form, we do not have a classification yet. Example~\ref{example:yellow} shows a tractable union of such form, while the following example shows that free-paths that share edges can be especially problematic within a union.

\begin{example}[Unresolved UCQ with more than two body-isomorphic CQs]\label{example:star}
Let $k\geq 4$ and consider the UCQ $Q$ containing $k$ body-isomorphic CQs, with an atom $R_i(x_i,z)$ for every $1\leq i\leq k-1$. The heads are all possible combinations of $k-1$ out of the $k$ variables of the query, $\{z,x_1,\ldots,x_{k-1}\}$.
In the case of $k=4$ we have the following UCQ:
\begin{align*}
Q_1(x_1,x_2,x_3),Q_2(x_1,x_2,z),Q_3(x_1,x_3,z),Q_4(x_2,x_3,z)\leftarrow R_1(x_1,z),R_2(x_2,z),R_3(x_3,z).
\end{align*}
The query $Q_1$ has free-paths $(x_i,z,x_j)$ between all possible pairs of $i$ and $j$.
In this case, enumerating the solutions to the UCQ is not in $\DelayClin$ assuming \fourclique:
Encode each relation with all edges in the input graph, and concatenate the variable names. That is, for every edge $(u,v)$ in the graph, add $((u,x_1),(v,z))$ to $R_1$. By concatenating variable names, we can identify which solutions come from which CQ as described in Lemma~\ref{lemma:diff-domain} and ignore the answers to all CQs other than $Q_1$.
The answers to $Q_1$ give us $3$ vertices that have a common neighbor. 
We can check in constant time if every pair of the $3$ vertices are neighbors, and if so, we found a $4$-clique. As there can be $O(n^3)$ solutions to the UCQ, we solve $4$-clique in $O(n^3)$ time.
This proves that this UCQ when $k=4$ is intractable assuming \fourclique.

If we apply the same strategy for a general $k$, we get that tractability of this UCQ implies solving $k$-clique in time $O(n^{k-1})$. However, this does not result in a lower
bound for large $k$ values: It is assumed that one can not find a $k$-clique
in time $n^{\frac{\omega k}{3}-o(1)}$, which does not contradict an $O(n^{k-1})$ algorithm.
In fact, for sufficiently large $k$, algorithms for $k$-clique detection that require time $O(n^{k-1})$ do exist~\cite{eisenbrand2004complexity}.
However, this reduction does not seem to fully capture the hardness of this query, as it encodes all relations with the same set of edges. 
We do not know if queries of the structure given here are hard in general, or if they become easy for larger $k$ values, as any current approach that we know of for constant delay enumeration fails for them.
\qed
\end{example}

We now describe some classification we can achieve when we exclude cases like Example~\ref{example:star}.

\subsubsection{Lower bound for isolated free-paths}

When generalizing the notion of guarding free-paths to unions of several CQs, a free-path does not have to be guarded by a single CQ.
We formalize this in the following definition.

\begin{definition}\label{definition:union-guard}
	Let $Q=Q_1\cup\ldots\cup Q_n$ be a union of body-isomorphic CQs, and
	let $P=(z_0,\ldots,z_{k+1})$ be a free-path in $Q_1$.
	We say that a set $\calU\subseteq 2^{\var(P)}$ is a {\em union guard} for $P$ if:
	\begin{itemize}
	\item $\{z_0,z_{k+1}\}\in\calU$.
	\item For every $\{z_a,z_c\}\subseteq u\in \calU$ with $a+1<c$, we have that $\{z_a,z_b,z_c\}\in\calU$ for some $a<b<c$.
	\item For every $u\in\calU$, we have $u\subseteq\free(Q_i)$ for some $1\leq i\leq n$.
	\end{itemize}
	We sometimes refer to the set $\{z_a,z_b,z_c\}$ with $a<b<c$
as $(z_a,z_b,z_c)$. 
\end{definition}

Note that, when $n=2$, if a UCQ is free-path guarded by Definition~\ref{def:two-guarded}, every free-path $P$ it contains has the union guard $\calU=\{\var(P)\}$.
Definition~\ref{definition:union-guard} allows several variable sets in the union guard, and different sets are allowed to be free in different CQs.
We now show that if a UCQ contains a free-path with no union guard, then the entire union is intractable.

\begin{theorem}\label{thm:no-union-guard}
	Let $Q=Q_1\cup\ldots\cup Q_n$ be a UCQ of body-isomorphic acyclic CQs.
	If there exists a free-path in $Q_1$ that is not union guarded, then
	$\pEnum{Q}\not\in\DelayClin$, assuming \matmul.
\end{theorem}
\begin{proof}
Let  $P=(z_0,\ldots,z_{k+1})$ be a free-path of $Q_1$ that is not union guarded. 
We know that $\{z_0,z_{k+1}\}\subseteq\free(Q_1)$. Since $P$ is not union guarded, there exist some $a$ and $c$ such that  $a+1< c$, and $\{z_a,z_c\}\subseteq \free(Q_r)$ for
some $Q_r\in Q$, but
for all $Q_s\in Q$ and for all $a<b<c$ we have $\{z_a,z_b,z_c\}\not\in\free(Q_s)$.
Note that $P'=(z_a,z_{a+1},\ldots,z_c)$ is a free-path of $Q_r$.
We define $V_a=\{z_a\}$, $V_b=\{z_{a+1},\ldots,z_{c-1}\}$ and $V_c=\{z_c\}$.
Since $P'$ is a free-path, there is no atom in $Q$ that contains both $z_a$ and $z_c$ as variables.

The rest of the proof is similar to Lemma~\ref{lemma:guardedTwoQueries}.
We partition the atoms into nonempty sets
$\calR_A=\{R(\vec{v})\in\atoms(Q)\mid V_c\cap \vec{v}=\emptyset\}$ and $\calR_B=\atoms(Q)\setminus\calR_A$, and we have that the atoms of $\calR_A$ do not contain the variable of $V_c$, and the atoms of $\calR_B$ do not contain the variable of $V_a$.
Let $A$ and $B$ be Boolean $n\times n$ matrices represented as binary relations.
We define a database instance $I$ over $Q$ as follows: For every atom $R(\vec{v})$, if $R(\vec{v})\in \calR_A$ we set
$R^I = \{\tau_{(a,b,\bot)}(\vec{v})\mid (a,b)\in A\}$,
and if $R(\vec{v})\in \calR_B$ we set 
$R^I = \{\tau_{(\bot,b,c)}(\vec{v})\mid (b,c)\in B\}$. Here, $\tau_{(a,b,c)}$ is the function defined in Lemma~\ref{lemma:guardedTwoQueries}.

Consider an answer $\mu\in Q(I)$.
Since the variables $z_i$ are connected via the path in all CQs,
we have that $\mu(z_a)=a$,
$\mu(z_{a+1})=\cdots=\mu(z_{c-1})=b$ and $\mu(z_{c})=c$ for some $(a,b)\in A$ and $(b,c)\in B$.
Since $\{z_a,z_c\}\subseteq \free(Q_r)$, the product $AB$ is encoded in $Q_r(I)$.
 Since for all $Q_s\in Q$ and for all $a<b<c$ we have $\{z_a,z_b,z_c\}\not\in\free(Q_s)$, every
$Q_s(I)$ is of size at most $\calO(n^2)$.
Thus there are $O(n^2)$ answers to $Q$.
Assume by contradiction that we can enumerate the solutions of $Q(I)$ with linear preprocessing
and constant delay.
To distinguish the answers of $Q_r$ from those of the other CQs, we can concatenate the variable names to the values, as described in Lemma~\ref{lemma:diff-domain}.
That way, we can ignore the answers that correspond to the values $(a,b)$ or $(b,c)$, and use the $(a,c)$ pairs as the answers to
matrix multiplication. This solves matrix multiplication in $O(n^2)$ time, in contradiction to \matmul.
\end{proof}

As we do not know of a classification for Example~\ref{example:star}, we restrict the UCQs we consider to cases where the free-paths within each CQ do not share variables. Then, we manage to obtain a similar characterization to that of Section~\ref{sec:2-body-iso}.

\begin{definition}\label{definition:isolated}
Let $Q=Q_1\cup\ldots\cup Q_n$ be a union of body-isomorphic CQs,
and let $P$ be a free-path of $Q_1$. We say that $P$ is {\em isolated} if the following two conditions hold:
\begin{itemize}
\item $Q$ is $\var(P)$-connex and
\item $\var(P')\cap\var(P)=\emptyset$ for all free-paths $P'\neq P$ in $Q_1$ .
\end{itemize}
\end{definition}

Note that isolated is a stronger property than bypass-guarded. Given a free-path $P$, A bypass-guarded query can have a variable in two subsequent $P$-atoms as long as this variable is free in another CQ. An isolated free-path cannot have such a variable at all.
If all free-paths in a union of body-isomorphic acyclic CQs are union guarded and isolated, we can show that the union is tractable.
To achieve this result, we first need to show a structural property given in the following lemma.

\begin{lemma}\label{lemma:guardedandisolated}
Let $Q=Q_1\cup\ldots\cup Q_n$ be a UCQ, $P=(z_0,\ldots,z_{k+1})$ a free-path of $Q_1$ and
$\calU$ a union guard of $P$. There exists some $\calU'\subseteq\calU$ such that:
\begin{itemize}
\item $\{z_0,z_{k+1}\}\subseteq v\in\calU'$.
\item for all $1\leq i\leq m$, $\{z_{i-1},z_i\}\subseteq v\in\calU'$.
\item there exists a join-tree $T_P$ for $\calH=(\var(P),\calU')$.
\end{itemize}
\end{lemma}
\begin{proof}
We use the inductive definition of a union guard to define $\calU'$ and $T_P$. 
Moreover, for the ease of explanations, we define $T_P$ with a parent-child relation. 
By the first two points of Definition~\ref{definition:union-guard},
$\{z_0,z_{k+1}\}\in\calU$ and thus we have $(z_0,z_\ell,z_{k+1})\in\calU$ for some
$1\leq \ell\leq k$. We set $(z_0,z_\ell,z_{k+1})$ as a root vertex of $T_P$.
The second condition of
Definition~\ref{definition:union-guard} defines the parent-child condition
of $T_P$ with at most two children per vertex: If $v=(z_a,z_b,z_c)\in V(T_P)$, then
	\begin{itemize}
	\item if $b> a+1$,  there exists some $\{z_a,z_{j},z_b\}\in\calU$ with $a<j<b$. Add one such
	vertex to the set of children of $v$.
	\item if $c> b+1$, there exists some $\{z_b,z_{j},z_c\}\in\calU$ with $b<j<c$.
	Add one such vertex to the set of children of $v$.
	\end{itemize}
Note that given $\calU$, the choice of $T_P$ might not be unique, and
the vertices in $T_P$ correspond to a subset $\calU'$ of $\calU$.

\begin{claim}
$T_P$ is a join tree of $\calH$.
\end{claim}
\begin{proof}[Proof of the Claim]
For every $v\in V(T_P)$, denote by $T_P(v)$ the maximal subtree of $T_P$ rooted in $v$.
Note that for $v=(z_a,z_b,z_c)$, we have $\var(T_P(v))\subseteq \{z_a,\ldots,z_c\}$
by construction.

We first prove that for every $v_p,v_c\in V(T_P)$, if $v_c$ is a descendent of $v_p$
and $z_i\in v_c\cap v_p$, then $z_i$ is contained in every vertex on the path between $v_p$ and
$v_c$.
Let $v_c\in T_P(v_p)$, and by way of contradiction let $v_q$ be 
the first vertex the path not containing $z_i$. Then by construction of $T_P$, the parent of $v_q$ must
either be of the form $(z_i, z_a, z_b)$, or $(z_a, z_b,z_i)$.
In both cases, $v_q$ is of the form $(z_a, z_j, z_b)$, and $\var(T_P(v_q))\subseteq\{z_a,\ldots, z_b\}$.
In the first case $a>i$, and in the second case $b<i$. In either case $z_i\not\in T_P(v_q)$, which is a contradiction to the fact that $z_i\in v_c$.

Let $v_1,v_2$ be two distinct vertices in $T_P$ and $z_i\in v_1\cap v_2$ for some $0\leq i\leq k+1$.
Let $v_3$ be a vertex on the unique path from $v_2$ to $v_1$.
We make a case distinction by how $v_1$ and $v_2$ are connected.
In case either $v_1\in T_P(v_2)$ or $v_2\in T_P(v_1)$, we already showed that $v_3$ contains $z_i$.
So assume that $v_1\not\in T_P(v_2)$ and $v_2\not\in T_P(v_1)$. This means that there is some 
$v_p\in V(T_P)\setminus\{v_1,v_2\}$ with $v_1,v_2\in T_P(v)$, and distinct children $u_1,u_2$ of $v_p$ such that
$v_1\in T_P(u_1)$ and $v_2\in T_P(u_2)$. For $v_p=(z_a, z_b, z_c)$ we have that
$\var(T_P(u_1))\cap\var(T_P(u_2))=\{z_b\}$ by construction, and since $z_i\in v_1\cap v_2$
it follows that $z_b=z_i$. Thus, either
$v_3$ is on the unique path from $v_p$ to $v_1$ or from $v_p$ to $v_2$. Since $v_1$ and $v_2$ are descendents of $v_p$, we already showed that $z_b\in v_3$.
\end{proof}

\begin{claim}
For every $1\leq i\leq k+1$, the set $\{z_{i-1},z_i\}$ is contained in a vertex of $T_P$.
\end{claim}
\begin{proof}[Proof of the Claim]
We describe a path from the root to a vertex containing $\{z_{i-1},z_i\}$.
We start with the root, and at each step we consider a vertex $v=(z_a,z_b,z_c)$ such that $a\leq i-1,i \leq c$.
We have that either $i\leq b$ or $i-1\geq b$. Assume that $i\leq b$.
If $b=a+1$, we have that $a=i-1$ and $b=i$, so we found the vertex we need.
Otherwise, $v$ has a child $(z_a,z_t,z_b)$ with $a\leq i-1,i \leq b$.
We consider this child next. The case that $i-1\geq b$ is symmetrical.
Since the tree is finite, the process will end and we will find such vertex.
\end{proof}

This concludes the proof of the Lemma.
\end{proof}

With this structural property at hand, we can show a union-extension for CQs in which every free-path is union guarded and isolated.

\begin{theorem}\label{thm:isolated}
Let $Q=Q_1\cup\ldots\cup Q_n$ be a union of body-isomorphic acyclic CQs.
If every free-path in $Q$ is union guarded and isolated, then
$\pEnum{Q}\in\DelayClin$.
\end{theorem}
\begin{proof}
Let $Q_1\in Q$. We use the fact that every free-path in $Q_1$ is union guarded and isolated to show that $Q_1$ is union-free-connex
with respect to $Q$. Since this is true for every CQ in the union, $Q$ is free-connex, and according to Theorem~\ref{thm:positive}, $\pEnum{Q}\in\DelayClin$. 

We show how to eliminate a free-path in $Q_1$ by adding virtual atoms. This process can be applied repeatedly to treat all free-paths in $Q_1$.
Let $P=(z_0,\ldots,z_{k+1})$ be a free-path in $Q_1$, and consider a join-tree $T_P$ given by Lemma~\ref{lemma:guardedandisolated}.
We extend every CQ $Q_j$ in $Q$ to a union extension $Q_j^+$ as follows: For every $v\in T_P$, add the atom $R_v(\vec{v})$ to $\atoms(Q)$, where $R_v$ is a fresh relational symbol. 
In Claim~\ref{claim:unionextension} we show that these variables sets are supplied, so this is indeed a valid union extension. 
In Claim~\ref{claim:famousclaim4} we show that extension eliminates $P$ without introducing new free-paths.
As the proof of Claim~\ref{claim:unionextension} uses a bottom-up induction on $T_P$, we will need to use Claim~\ref{claim:acyclicprovided} regarding the subtrees of $T_P$. We use the same notation as in the proof of Lemma~\ref{lemma:guardedandisolated}: For every vertex $v\in V(T_P)$, we denote by $T_P(v)$ the subtree of $T_P$ rooted in $v$. 
Note that $V(T_P(v))$ is a subset of $\calU'$.

\begin{claim}\label{claim:acyclicprovided}
Let $v\in V(T_P)$.
There exists an ext-$\var(T_P(v))$-connex
tree $T'$ for the hypergraph $\calH'=(\var(Q),E(\calH(Q))\cup V(T_P(v)))$. 
\end{claim} 
\begin{proof}[Proof of the Claim]

Let $v=(z_a,z_b,z_c)$. By construction of $T_P$, we
have that $\var(T_P(v))=\{z_a,z_{a+1},\ldots,z_c\}$. Denote by $R$ the
subpath $(z_a,z_{a+1},\ldots,z_c)$ of $P$.
It is possible to show that since $P$ is a chordless path, for every subpath of $P$, we have that $Q$ is $\var(R)$-connex.
Let $T$ be an ext-$\var(R)$-connex tree for $\calH(Q)$, and
let $T_R\subseteq T$ be a connected subtree
of $T$ with $\var(T_R)=\var(R)$. We construct a new tree
$T'$ by first removing the edges among vertices of $T_R$, adding the tree $T_P(v)$ and
then reconnecting the vertices of $T_R$ to $T_P(v)$ as fellows:
Let $u\in V(T_R)$. Since $R$ is a chordless path, we have that $u=\{z_s\}$ for some $a\leq s\leq c$
or $u=\{z_{s},z_{s+1}\}$ for some $a\leq s\leq c-1$. Thus there exists some vertex
$w\in T_P(v)$ with $u\subseteq w$. Chose some arbitrary $w\in T_P(v)$ with this property and
add an edge $(u,w)$. 

Since $T$ is a tree, removing the edges of $T_R$ results in a forest.
For every $u\in V(T_R)$, the connected component of this forest
that contains $u$ does not contain any other vertices in $V(T_R)$. Thus in every step of adding
an edge, we attach a new tree to $T_P(v)$, and no such tree is attached to $T_P(v)$ more then
once. Thus $T'$ is again a tree and acyclic.
\end{proof}

\begin{claim}\label{claim:unionextension}
$Q_1^+$ is a union-extension.
\end{claim}
\begin{proof}
We prove via a bottom-up induction on $T_P$ that every vertex $v\in V(T_P)$ is supplied by some union-extension of a CQ in $Q$.

For the base case, note that the leaves
of $T_P$ are triples of the form $(z_i,z_{i+1},z_{i+2})$ for $0\leq i\leq k-1$.
Since $P$ is a chordless path, for every subpath of $P$ we have that $Q$ is $\var(R)$-connex.
Therefore, $Q$ is $\{z_i,z_{i+1},z_{i+2}\}$-connex. Since also $(z_i,z_{i+1},z_{i+2})$ is contained in some $\free(Q_j)$, all leaves are supplied. 

Consider some vertex $v=(z_a,z_b,z_c)$ of $T_P$ that is not a leaf, and assume that for every child $v'$ of $v$, we have that every vertex in $T(v')$ is supplied. 
We need to show that $v$ is supplied.
Consider the ext-$\var(T_P(v))$-connex tree $T'$ of the hypergraph $(\var(Q),E(\calH(Q))\cup V(T_P(v)))$, 
that was constructed in the proof for Claim~\ref{claim:acyclicprovided}. In this tree, we replace
the node $v$ by the two nodes $v_1=\{z_a,z_b\}$ and $v_2=\{z_b,z_c\}$, and then add the edge
$(v_1,v_2)$. For every edge $e=(u,v)$ that was lost when deleting $v$, we do the following: If $u$ contains the variable $z_c$, add an edge $(u,v_2)$, otherwise add the edge $(u,v_1)$. Since no vertex in $T'$ besides $v$ contains both $z_a$ and $z_c$, this is again a valid join tree. Moreover,
both $\{z_a,z_b\}$ and $\{z_b,z_c\}$ are contained in vertices of $T'\setminus \{v_1,v_2\}$:
If $v$ has two children, then the children contain $\{z_a,z_b\}$ and $\{z_b,z_c\}$. Otherwise,
if $v$ has only one child node, we have that either $b=a+1$ or $c=a+1$. In both such cases,
$\{z_a,z_b\}$ or $\{z_b,z_c\}$ is an edge of the path $P$ and thus contained in a vertex in $T'$.
Thus we have that $T'$ is
ext-$\{z_a,z_b,z_c\}$-connex acyclic. Let $Q_j\in Q$ such that $\{z_a,z_b,z_c\}\subseteq \free(Q_j)$.
Since every vertex in $T'$ is supplied, there exists a union-extension $Q_j^+$ with
$Q_j^+(\free(Q))\leftarrow R_{v_1}(\vec{v_1}),\ldots,R_{v_N}(\vec{v_N})$
and $\{v_1,\ldots,v_N\}=V(T')$. Therefore, $Q_j^+$
supplies $v$.
\end{proof}

\begin{claim}\label{claim:famousclaim4}
The set of free-paths in $Q_1^+$ equals the set of free-paths in $Q_1$ minus $P$.
\end{claim}
\begin{proof}[Proof of the claim]
Since $\{z_0,z_{k+1}\}$ is contained in a vertex of $T$ by  Lemma~\ref{lemma:guardedandisolated},
it is also contained in the variables of some added atom $R_v(\vec{v})$, thus $P$ is not
a free-path of $Q_1^+$. For the sake of a contradiction, assume that there is some new free-path
$P'=(z_0',\ldots,z_{m+1}')$ in $Q_1^+$. The only new edges in $\calH(Q_i^+)$ are between variables in 
$\var(P)$, thus $|var(P)\cap\var(P')|\geq 2$. Let $z_s$ and $z_t$ be the first and last elements in $P'$ that are also in $P$. Note that the $z_s$ and $z_t$ are not neighbors in $P$.
Replacing the path between $z_s$ and $z_t$ in $P'$ with the path 
between $z_s$ and $z_t$ in $P$, we get the path
$P''=(z_0',\ldots, z_s, z_{s+1},\ldots,z_{t-1},z_t,\ldots,z_{m+1}')$,
which is a path in $\calH(Q)$. 
This path contains a chordless sub-path $P'''$ between $z_0'$ and $z_{m+1}'$ containing at least one element in $\{z_s,z_{s+1},\ldots,z_{t-1},z_t\}$. Thus $P'''$ is a free-path in $Q_1$ with a non-empty intersection with $P$, which is a contradiction to the assumption that every free-path is isolated.
\end{proof}

If $Q_1^+$ is free connex, then we are done. Otherwise, by Claim~\ref{claim:famousclaim4},
every free-path $P'$ in $Q_1^+$ is a free-path in $Q_1$ and thus union guarded and isolated in $Q_1$. Since the added atoms only contain variables of $\var(P)$ and $\var(P)\cap\var(P')=\emptyset$, we have that $P'$ is also union guarded and isolated in $Q_1^+$.
The tree-structure from Claim~\ref{claim:acyclicprovided} over the root is a join-tree for the union extension.
Since the extension consists of body-isomorphic acyclic queries, we can iteratively apply this process until $Q_1$ becomes free-connex.
\end{proof}

Theorem~\ref{thm:no-union-guard} and Theorem~\ref{thm:isolated} form a dichotomy for unions of body-isomorphic acyclic CQs where all free-paths are isolated: if all free-paths are are guarded, the UCQ is in $\DelayClin$; otherwise, it is not in $\DelayClin$ assuming $\matmul$. 
It is still left to handle cases like Example~\ref{example:star}, of unions containing body-isomorphic acyclic CQs where some free-path is union guarded but not isolated. After solving this case, and regarding unions of difficult CQs, we should also explore unions containing body-isomorphic acyclic CQs but also other difficult CQs.

\subsection{Unions Containing Cyclic CQs}\label{sec:future-cyclic}

We now discuss UCQs containing at least one cyclic query. 
We first mention that many of the observations we have regarding acyclic CQs also apply here.
In the following examples, $Q_1$ is cyclic while $Q_2$ is free-connex. Example~\ref{example:cyclic-easy} shows that unions containing cyclic CQs may be tractable and covered by Theorem~\ref{thm:positive}. Example~\ref{example:cyclic-free-path} demonstrates that it is not enough to resolve the cyclic structures in CQs, but that we should also handle free-paths. Finally, Example~\ref{example:cyclic-guarded-hard} shows that, much like Example~\ref{example:separated} in the acyclic case, even if all difficult structures are unguarded, the original reductions showing the intractability of single CQs may not work.

\begin{example}[A tractable UCQ containing a cyclic CQ]\label{example:cyclic-easy}
Let $Q=Q_1\cup Q_2$ with
\begin{align*}
	Q_1(x,y,z,w)\leftarrow& R_1(y,z,w,x),R_2(t,y,w),R_3(t,z,w),R_4(t,y,z),\\
	Q_2(x,y,z,w)\leftarrow &R_1(x,z,w,v),R_2(y,x,w).
\end{align*}
The CQ $Q_2$ is free-connex, and $\{y,x,z,w\}\subseteq\free(Q_2)$.
Thus, $Q_2$ supplies $V_1$ to $Q_1$.
The body-homomorphism from $Q_2$ to $Q_1$ maps $V_2=\{y,x,z,w\}$ to $V_1=\{t,y,z,w\}$.
Adding the virtual atom $R'(t,y,z,w)$ to $Q_1$ results in a free-connex union-extension, so $\pEnum{Q_1\cup Q_2}\in\DelayClin$ according to Theorem~\ref{thm:positive}.
\qed
\end{example}

\begin{example}[An intractable UCQ with guarded cycles]\label{example:cyclic-free-path}
Let $Q=Q_1\cup Q_2$ with
\begin{align*}
	Q_1(x,y,v)&\leftarrow R_1(v,z,x),R_2(y,v),R_3(z,y)\text{ and}\\
	Q_2(x,y,v)&\leftarrow R_1(y,v,z),R_2(x,y).
\end{align*}
The union $Q_1\cup Q_2$ is intractable despite the fact that $Q_2$ guards the cycle variables $\{v,y,z\}$. This is due to the unguarded free-path $(x,z,y)$ in $Q_1$. Similarly to Example~\ref{example:body-iso-matmul}, we can encode matrix multiplication to $x,z,y$ and conclude that $\pEnum{Q}\not\in\DelayClin$ assuming $\matmul$.
\qed
\end{example}

\begin{example}[An unresolved UCQ with unguarded difficult structures]\label{example:cyclic-guarded-hard}
Let $Q=Q_1\cup Q_2$ with
\begin{align*}
	Q_1(x,z,y,v)&\leftarrow R_1(x,z,v),R_2(z,y,v),R_3(y,x,v)\text{ and}\\
	Q_2(x,z,y,v)&\leftarrow R_1(x,z,v),R_2(y,t_1,v),R_3(t_2,x,v).
\end{align*}
We do not know the complexity of this example.
As $Q_2$ does not have a free variable that maps to $y$ via a homomorphism, 
$Q_1$ is not union-free-connex, so we cannot conclude using Theorem~\ref{thm:positive} that $Q_1\cup Q_2$ is tractable.
The only difficult structure in $Q_1$ is the cycle $x,y,z$, but encoding the triangle finding problem to this cycle in $Q_1$, like we did in Example~\ref{example:intractables}, could result in $n^3$ answers to $Q_2$. This means that if the input graph has triangles, we are not guaranteed to find one in $O(n^2)$ time by evaluating the union efficiently.
\qed
\end{example}

In addition to these issues, in the cyclic case, even if the original difficult structures are resolved, the extension may introduce new ones. Resolving the following example in general is left for future work.

\begin{example}[An unresolved UCQ with guarded difficult structures]\label{example:newtetra}
We start with a specific UCQ that we later generalize into a family of UCQs.
Let $Q=Q_1\cup Q_2$ with
\begin{align*}
	Q_1(x_2,x_3,x_4)&\leftarrow R_1(x_2,x_3,x_4),R_2(x_1,x_3,x_4),R_3(x_1,x_2,x_4),\\
	Q_2(x_2,x_3,x_4)&\leftarrow R_1(x_2,x_3,x_1),R_2(x_4,x_3,v).
\end{align*}
There is a body-homomorphism from $Q_2$ to $Q_1$, but $Q_1$ is not contained in $Q_2$. 
$Q_1$ is cyclic, as it has the cycle $(x_1,x_2,x_3)$. This is the only difficult structure in $Q_1$, i.e. $\calH(Q_1)$ does not contain a hyperclique or a free-path. The query $Q_2$ is both free-connex and $\{x_2,x_3,x_4\}$-connex, and so it supplies $\{x_2,x_3,x_4\}$. These variables are mapped to $\{x_1,x_2,x_3\}$ via the body-homomorphism to $Q_1$. Nevertheless, extending $Q_1$ with a virtual atom $R(x_1,x_2,x_3)$ does not result in a free-connex CQ.
Even though the extension ``removes'' all difficult structures from $Q_1$, it is intractable as it introduces a new difficult structure, namely the hyperclique $\{x_1,x_2,x_3,x_4\}$.

In this case, we have $\pEnum{Q_1\cup Q_2}\not\in\DelayClin$ assuming \fourclique, with
a reduction similar to that of Example~\ref{example:acyclic-4clique}:
Given an input graph, compute all triangles and encode them to the three relations. For every triangle $\{a,b,c\}$, add 
the tuple $((a,x_2),(b,x_3),(c,x_4))$ to $R_1$, $((a,x_1),(b,x_3),(c,x_4))$ to $R_2$ and $((a,x_1),(b,x_2),(c,x_4))$ to $R_3$. By concatenating variable names, we can identify which solutions correspond to which CQ as described in Lemma~\ref{lemma:diff-domain}, and thus are able to ignore the answers to $Q_2$.
The answers to $Q_1$ represent $3$ vertices that appear in a $4$-clique:
For every answer $((b,x_2),(c,x_3),(d,x_4))$ to $Q_1$, we know that $((b,x_2),(c,x_3),(d,x_4))\in R_1$,
and there exists some $a$ with $((a,x_1),(c,x_3),(d,x_4))\in R_2$ and $((a,x_1),(b,x_2),(d,x_4))\in R_3$.
This means that $\{a,b,c,d\}$ is a $4$-clique.
In the opposite direction, for every $4$-clique $\{a,b,c,d\}$ we have that $\{b,c,d\}$, $\{a,c,d\}$ and
$\{a,b,d\}$ are triangles. By construction, the tuple $((b,x_2),(c,x_3),(d,x_4))$ is an answer to $Q_1$.
As there are $O(n^3)$ triangles and there can be at most $O(n^3)$ solutions to $Q_2$, we solve $4$-clique in $O(n^3)$ time.
This proves that the UCQ is intractable assuming \fourclique.

This example can be generalized to higher orders. There, we do not have a similar lower bound. Consider the union of the following:
\begin{align*}
	Q_1(x_2,\ldots,x_k)\leftarrow& \{R_i (x_1,...,x_{i-1},x_{i+1},...,x_k)\mid 1\leq i\leq k-1\}\\
	Q_2(x_2,\ldots,x_k)\leftarrow& R_1(x_2,\ldots,x_{k-1},x_1),R_2(x_k,x_3,\ldots,x_{k-1},v).
\end{align*}
Again, the query $Q_1$ is cyclic and $Q_2$ is free-connex. Even though $Q_2$ supplies variables that map to $\{x_1,\ldots,x_{k-1}\}$ via a body-homomorphism, adding a virtual atom with these variables does not result in a free-connex extension, as this extension is again cyclic.
As in Example~\ref{example:star}, we can encode $k$-clique to this example in general, but this does not imply a lower bound for large $k$ values.
\qed
\end{example}

\section{Modified Settings}\label{sec:extensions}

We devote this section to discussing the implications of this work when the settings are slightly different than the ones assumed so far:
in Section~\ref{sec:CDs} we discuss the common case that there are cardinality dependencies in the database schema; in Section~\ref{sec:dis} we discuss unions of conjunctive queries with disequalities; and in Section~\ref{sec:space} we discuss the case where we want to restrict the space requirements of our algorithm.
Reaching complete results regarding the tractability of UCQs in these settings is beyond the scope of this work, but we do discuss some insights we can conclude regarding these settings following our work.

\subsection{Cardinality Dependencies}\label{sec:CDs}

In this work, we inspect the complexity of answering UCQs when we cannot assume anything on the database instance given as input. Nevertheless, usually we do know what the database we have at hand represents, and this implies restrictions on the combinations of values that the relations may hold.
One way of formalizing such restrictions is through Cardinality Dependencies (CDs), also known as degree bounds~\cite{abo2017shannon}, given as part of the schema. Cardinality dependencies are of the form $(R:A\rightarrow B, c)$ where $A,B\subseteq\{1,\ldots,\arity(R)\}$ and $c$ is a positive natural number. This dependency means that for every database instance of this schema, the number of tuples in $R$ with the same values in the $A$ indices but distinct values in the $B$ indices is at most $c$. Functional Dependencies (FDs) are a special case of cardinality dependencies where $c=1$, and keys are a special case of FDs where $A\cup B = \{1,\ldots,\arity(R)\}$.

As we know, the tractable CQs over general instances are the ones with a free-connex structure. In the presence of cardinality dependencies, additional CQs are tractable. The tractable CQs can be identified by treating the dependencies as between variables and ``chasing'' the dependencies: the variables of the right-hand side of a dependency are added to any atom that contains the variables of its left-hand side; this extension is applied to the head of the CQ in addition to the atoms. Applying any such number of extension steps generates a \emph{CD-extended CQ}, and if this CQ is free-connex, the original CQ is tractable~\cite{DBLP:conf/icdt/CarmeliK18}.
This resembles the technique used in this article to show the tractability of a UCQ where we extend the CQs of the union with additional atoms to reach a tractable (free-connex) structure.
Following is the formal definition of CD-Extended CQs.

\begin{definition}[CD-Extended CQ~\cite{DBLP:conf/icdt/CarmeliK18}]
	Let $Q(\vec{p}) \leftarrow R_1(\vec{v_1}), \dots, R_m(\vec{v_m})$ be a CQ
	over a schema with the dependencies $\Delta$. 
	There are two possible types of extension steps:
	\begin{itemize}
		\item The extension of an atom $R_i(\vec{v_i})$ according to an CD $R_j\colon A\rightarrow b$.\\
		Prerequisites: $\vec{v_j}[A]\subseteq \vec{v_i}$ and $\vec{v_j}[b]\notin \vec{v_i}$.\\
		Effect: The arity of $R_i$ increases by one, $R_i(\vec{v_i})$ is replaced by $R_i(\vec{v_i},\vec{v_j}[b])$.
		In addition, every $R_k(\vec{v_k})$ such that $R_k$=$R_i$ and $k\neq i$ is replaced with $R_k(\vec{v_k},t_k)$, where $t_k$ is a fresh variable in every such step.
		\item The extension of the head $Q(\vec{p})$ according to an CD $R_j\colon A\rightarrow b$.\\
		Prerequisites: $\vec{v_j}[A]\subseteq \vec{p}$ and $\vec{v_j}[b]\notin \vec{p}$.\\
		Effect: The head is replaced by $Q(\vec{p},\vec{v_j}[b])$.
	\end{itemize}
	The {\em CD-extension} of $Q$ is the CQ obtained by performing all possible extension steps on $Q$ according to CDs of $\Delta$ until a fixpoint is reached.
\end{definition}

We devote this section to discussing whether there is a need to combine these two techniques and how this can be done.
First, we discuss the need to combine the two techniques.

\begin{example}\label{example:FD}
Consider the UCQ $Q=Q_1\cup Q_2$ with:
\begin{align*}
	Q_1(x,w,u)\leftarrow& R_1(x,y),R_2(y,z),R_3(z,w),R_4(u)\\
	Q_2(x,w,u)\leftarrow& R_1(x,w),R_2(t,u)
\end{align*}
over a schema with the FD $R_3:2\rightarrow 1$.
\end{example}

The CQ $Q_1$ in itself is not free-connex, as it contains the free-path $(x,y,z,w)$.
Even if we take the FD into account, $Q_1$ is not tractable.  The FD-extension of $Q_1$ is obtained by "chasing" the FD $w\rightarrow z$, and adding $z$ to the query head. 
$$
Q_1^+(x,w,u,z)\leftarrow R_1(x,y),R_2(y,z),R_3(z,w),R_4(u)
$$
$Q_1^+$ is not free-connex as it contains the free-path $(x,y,z)$.
This means that by using the FD techniques alone, $Q$ does not become tractable.
If we take the union into account but ignore the FD, $Q_2$ can supply $\{x,w,u\}$, which map to $\{x,y,z\}$ in $Q_1$ via a body-homomorphism, but adding atoms with any subset of these variables cannot make $Q_1$ tractable, because it cannot prevent the appearance of a free-path between $x$ and $w$. It can only shorten it from $(x,y,z,w)$ to $(x,z,w)$. This means that by using the union techniques alone, $Q$ does not become tractable.
Nonetheless, if we combine the extensions, we both add $z$ to the head and add an atom with $\{x,y,z\}$, we get the following tractable extension:
\begin{align*}
	Q_1^{++}(x,w,u,z)\leftarrow& R_1(x,y),R_2(y,z),R_3(z,w),R_4(u),R'(x,y,z)
\end{align*}
Practically, this means that $Q$ can be solved efficiently by first solving $Q_2$, then using its results as another atom $R'$, then solving $Q_1^{++}$ with the help of $R'$, and translating every answer to $Q_1^{++}$ to an answer to $Q_1$ by projecting out $z$. The FD guarantees that the translation phase does not result in duplicates. If we had a cardinality dependency instead of an FD, we would have a constant number of duplicates per answer to $Q_1$, and the duplicates could be ignored efficiently using the Cheater's Lemma (Lemma~\ref{lemma:cheaters}).

Example~\ref{example:FD} shows that if we have UCQs over a schema with CDs, we can identify additional tractable cases by combining the two techniques.
The two types of extensions can easily be combined in any order. To show this, we just need to modify the definitions of a free-connex union to account for CDs, and use the same proof as before.
\begin{definition}
	Let $Q= Q_1\cup\ldots\cup Q_n$ be a UCQ over a schema with cardinality dependencies.
	\begin{itemize}
	\item A \emph{CD-extension} for $Q$ is of the form $Q= Q_1'\cup\ldots\cup Q_n'$ such that for all $i=1,\dots,n$ either $Q_i'=Q_i$ or $Q_i'$ is the CD-extension of $Q_i$.
		\item 
		A \emph{CD-union extension sequence} for $Q$ is a sequence $Q^1,\ldots,Q^N$ where $Q=Q^1$ and for all $1< j \le N$, it holds that $Q^j$ is either a union extension or a CD-extension of $Q^{j-1}$.
		  \end{itemize}
\end{definition}

Note that a CD-extension of a CQ may result in an increase in the arity of the head. Therefore, a CD-extension of a UCQ may not be a valid UCQ in terms of the standard definitions, as these require all CQ heads to have the same arity. However, as the algorithms we use in this article for the evaluation of UCQs do not rely on the size of the heads, we can evaluate such a union as usual.

\begin{theorem}
	Let $Q$ be a UCQ over a schema with cardinality dependencies. If $Q$ has a CD-union extension sequence ending in a UCQ containing only free-connex CQ, then $\pEnum{Q}\in\DelayClin$.
\end{theorem}
\begin{proof}
Consider an extension sequence $Q^1,\ldots,Q^N$. First, we inductively build an instance for $Q^N$. The instance for $Q^1$ is given as input.
If $Q^{i+1}$ is a union-extension of $Q^{i}$, use the construction given in this article for CQs without CDs (i.e., follow the provision phase described in Theorem~\ref{thm:positive}) to obtain the instance for $Q^{i+1}$.
Otherwise, $Q^{i+1}$ is a CD-extension of $Q^{i}$,
and so use the construction for CD-extensions~\cite[Theorem 2 (Claim 1)]{DBLP:conf/icdt/CarmeliK18} instead.
Then, evaluate every CQ in $Q^N$ using the CDY algorithm.
For every answer produced during the construction or in the final evaluation, translate it back to an answer to $Q^1$ by projecting out any extra variables (these were possibly introduced due to CD-extensions).
The correctness follows from the correctness of the individual extensions.
The desired time complexity is achieved using the Cheater's Lemma: since $N$ (the length of the extension sequence) is constant with respect to the input database, the number of times every answer is produced is bound by a constant.
\end{proof}

Note that in Example~\ref{example:FD} the order of applying the extension does not matter. That is, first adding $R'$ and then adding $z$ wherever $w$ appears results in the same query as first adding $z$ and then adding $R'$. This is not always the case.
Consider the same UCQ as before but with the FD $R_3:1\rightarrow 2$ instead of the reverse FD. When we treat the FD as between variables in $Q_1$, we get $z\rightarrow w$, and we can add $w$ wherever $z$ appears in $Q_1$ as an FD-extension. The individual extensions are still not enough:
adding $w$ to $R_2$ results in the free-path $(x,y,w)$, and adding an atom with $\{x,y,z\}$ or a subset of these variables does not help as before.
The combination of the two extensions is enough if we first use the union and only then use the FD. We get the following free-connex CQ.
\begin{align*}
	Q_1'(x,w,u)\leftarrow& R_1(x,y),R_2(y,z,w),R_3(z,w),R_4(u),R'(x,y,z,w)
\end{align*}
If we apply the extensions in the opposite order, we do not have $w$ in $R'$, and we have the free-path $(x,z,w)$, so the extension is not tractable.

Overall we saw that the two techniques can and should be combined to find more tractable cases. However, one should be aware that it is not enough to first take the CD-extensions of all CQs in the union and then perform union-extensions. The order in which the different extensions are combined matters.

\subsection{Disequalities}\label{sec:dis}
In this section we discuss the enumeration complexity of unions of CQs with disequalities. In the case of individual CQs with disequalities, the disequalities have no effect on the enumeration complexity~\cite{bdg:dichotomy}. That is, one can simply ignore the disequalities, and the remaining CQ is free-connex if and only if the query is tractable (under the same assumptions as in Theorem~\ref{theorem:originalDichotomy}).
A natural question is whether this happens also with UCQs. We next show that it does not, and some tractable UCQs become intractable with the addition of disequalitites.
\begin{example}
Let $Q=Q_1\cup Q_2$ with 
\begin{align*}
	Q_1(x,y,z,w)\leftarrow& R_1(x,y),R_2(y,z),R_3(z,z),R_4(x,z),R_5(w)\\
	Q_2(x,y,z,w)\leftarrow& R_1(x,y),R_3(w,z),z\neq w
\end{align*}
We can show this union is hard by encoding triangle detection to the cycle in $Q_1$ as usual. Due to the disequality, $Q_2$ has no answers over this construction, so finding an answer to the UCQ in linear time detects a triangle in linear time. If there was no disequality, $Q_2$ would supply variables that map to the cycle in $Q_1$ via a body-homomorphism, and the UCQ would become tractable.
\end{example}

This example shows that (unlike with CQs) the positive results presented in this article for UCQs do not hold when we ignore disequalities. Nevertheless, a simple adjustment of the proofs of Section~\ref{sec:positive} reveals cases where a UCQ with disequalities containing hard CQs is easy.

\begin{definition}\label{def:disequalities}
Let $Q_1,Q_2$ be CQs with disequalities.
\begin{itemize}
\item A \emph{body-homomorphism} from $Q_2$ to $Q_1$ is a mapping $h:\ivar(Q_2)\rightarrow\ivar(Q_1)$ such that:
\begin{itemize}
\item for every atom $R(\vec{v})$ of $Q_2$, $R(h(\vec{v}))\in Q_1$.
\item for every disequality $v \neq u$ of $Q_2$, $h(v)\neq h(u)$ is in $Q_1$.
\end{itemize}
\item We say that $Q_2$ \emph{supplies} a set $V$ of variables if $Q_2$ is free-connex and 
$V\subseteq\free(Q_2)$.
	\end{itemize}
\end{definition}

Note that the definition of a body-homomorphism now takes the disequalities into account.
Note also that the definition of supplying variables here is simpler than before as we require the supplying CQ to be free-connex.\footnote{This is a restriction we impose for simplicity, and it is possible that a more refined condition may lead to additional tractable cases. However, this seems to require going into the details of the algorithm for CQs with disequalities, and it is left for future work.}
We set the definition of a union extension as in Section~\ref{sec:positive} (but using the new definition of supplied variables and body-homomorphisms). Using these definitions, we can get a result similar to that of Theorem~\ref{thm:positive}.
We first prove the equivalent of Lemma~\ref{lemma:provide} in our settings.

\begin{lemma}\label{lemma:provide-dis}
Let $Q_2$ be a CQ with disequalities that supplies the variables $\vec{v}_2$.
Given an instance $I$,
one can compute with linear time preprocessing and constant delay $M=Q_2(I)$, and 
$M$ can be translated in time $O(|M|)$ to a relation $R^{M}$ such that:
for every CQ $Q_1$ with a body-homomorphism $h$ from $Q_2$ to $Q_1$
and for every answer $\mu_1\in\full(Q_1)(I)$,
there is the tuple $\mu_1(h(\vec{v}_2))\in R^{M}$.
\end{lemma}
\begin{proof}
Let $Q_2$ be a CQ with disequalities that supplies the variables $\vec{v}_2$.
Since $Q_2$ is free-connex, it can be answered with linear preprocessing and constant delay~\cite{bdg:dichotomy}.
We define $R^{M}=\{\mu_2(\vec{v}_2)\mid \mu_2\in Q_2(I)\}$.
Let $Q_1$ be a CQ such that there is a body-homomorphism $h$ from $Q_2$ to $Q_1$, and
let $\mu_1\in \full(Q_1)(I)$.
Since $h$ is a body-homomorphism, for every atom $R(\vec{v})$ in $Q_2$, $R(h(\vec{v}))$ is an atom in $Q_1$, and for every disequality $v \neq u$ of $Q_2$, $h(v)\neq h(u)$ is in $Q_1$.
Since $\mu_1$ is an answer to $\full(Q_1)$, for such atoms and disequalities $\mu_1(h(\vec{v}))\in R^I$ and $\mu_1(h(v))\neq \mu_1(h(u))$.
This means that $\mu_1\circ h|_{\free(Q_2)}$ is an answer to $Q_2$, so there exists $\mu_2\in Q_2(I)$ such that $\mu_1\circ h|_{\free(Q_2)}=\mu_2$.
By construction, $\mu_1(h(\vec{v}_2))=\mu_2(\vec{v}_2)\in R^{M}$.
\end{proof}

With Lemma~\ref{lemma:provide-dis} at hand, we can apply the proof of Theorem~\ref{thm:positive} as is. We only need to use Lemma~\ref{lemma:provide-dis} instead of Lemma~\ref{lemma:provide}. This proves the following result.

\begin{theorem}\label{thm:dis-positive}
	Let $Q$ be a union of CQ with disequalities. If $Q$ is free-connex, then $\pEnum{Q}$ is in $\DelayClin$.
\end{theorem}

\begin{example}\label{example:dis-positive}
Let $Q=Q_1\cup Q_2$ with 
\begin{align*}
Q_1(x,y,w)&\leftarrow R_1(x,z),R_2(z,y),R_3(y,w),x\neq z\text{ and }\\ 	
Q_2(x,y,w)&\leftarrow R_1(x,y),R_2(y,w),x\neq y
\end{align*}
This is the same as Example~\ref{example:first} except for the added disequalities. Without the disequalities this union is easy as $Q_2$ supplies $\{x,y,w\}$, and adding a virtual atom with a union extension to $Q_1$ results in a free-connex form. Since $Q_2$ is free-connex and the disequality of $Q_2$ maps to that of $Q_1$ using the body-homomorphism, this union is easy even with the disequalities according to Theorem~\ref{thm:dis-positive}.
\end{example}

Theorem~\ref{thm:dis-positive} can be used also to show the tractability of UCQs that are not naturally of that form.
For example, in the case of the last example, the union is easy even without the disequality in $Q_1$.
\begin{example}\label{ex:dis-also-pos}
Let $Q=Q_1\cup Q_2$ with 
\begin{align*}
Q_1(x,y,w)&\leftarrow R_1(x,z),R_2(z,y),R_3(y,w)\text{ and }\\ 	
Q_2(x,y,w)&\leftarrow R_1(x,y),R_2(y,w),x\neq y
\end{align*}
We can partition the answers to $Q_1$ to those that assign the same values to $x$ and $z$ and those that do not. $Q$ is equivalent to $Q'=Q_1^\neq\cup Q_1^=\cup Q_2$ with
\begin{align*}
Q_1^\neq(x,y,w)&\leftarrow R_1(x,z),R_2(z,y),R_3(y,w),x\neq z\text{ and }\\ 	
Q_1^=(x,y,w)&\leftarrow R_1(x,x),R_2(x,y),R_3(y,w)\text{ and }\\ 	
Q_2(x,y,w)&\leftarrow R_1(x,y),R_2(y,w),x\neq y
\end{align*}
Here, $Q_2$ and $Q_1^=$ are free-connex, and $Q_1^\neq$ has a free-connex union-extension as explained in Example~\ref{example:dis-positive}.
\end{example}

Unlike in the case of individual CQs, the complexity of answering unions of CQs changes with the addition of disequalities. We proved the tractability of some unions with disequalities that can be easily concluded from the results of this article, but identifying all of the tractable cases is left for future work.

\subsection{Restricted Space}\label{sec:space}

In this article, we only considered time bounds. The class $\CDlin$ describes the problems that can be solved with the same time bounds, but with the additional restriction that the available space for writing during the enumeration phase is constant. Evaluating free-connex CQs is in $\CDlin$, and Kazana offers a comparison between $\DelayClin$ and $\CDlin$~\cite[Section 8.1.2]{kazana:thesis}.
All lower bounds presented in this article naturally hold for $\CDlin$ as we know that, by definition of the classes, $\CDlin\subseteq\DelayClin$. The tractability of unions containing only free-connex CQs, as explained in Theorem~\ref{thm:easyunions}, also holds for this class.
However, in our techniques for the tractable UCQs that do not contain only tractable CQs, 
the memory to which we write during the enumeration phase may increase in size by a constant with every new answer.
An interesting question is whether we can achieve the same time bounds when restricting the memory according to $\CDlin$.

The first reason we use polynomial space is to regularize the delay and avoid duplicates. This is achieved in the Cheater's Lemma (Lemma~\ref{lemma:cheaters}) by storing all results we produce. 
We should mention in that context that regularizing the delay is arguably not of practical importance. In practice, once we have an answer, we will probably want to use it immediately.
For this reason, one might be satisfied with simply using a relaxed complexity measure which captures the fact that an algorithm performs "just as good" as linear preprocessing and constant delay when given enough space: linear partial time (see Definition~\ref{def:linpartial}).
The Cheater's Lemma shows that, whenever we have an algorithm $\calA$ that runs in linear partial time, there is a linear preprocessing and constant delay algorithm $\calA'$ that uses additional space and computes the same results, where every result in $\calA$ is returned no later than the same result in $\calA'$.
In fact, if we denote by $\linpartial$ the class of enumeration problems that can be solved by a linear partial time algorithm, $\linpartial = \DelayClin$.
Please note that this does not imply $\linpartial = \CDlin$, as we currently do not know whether $\DelayClin=\CDlin$.

The other reason that we use large amounts of space is that we need to store the generated relations that correspond to the virtual atoms in order to evaluate the union extension as specified in Theorem~\ref{thm:positive}.
The discussion above implies that it makes sense to focus on eliminating this fundamental reason.
As the next example demonstrates, this seems to require a different approach than the one presented in this article.

\begin{example}\label{example:no-space}
Let $Q=Q_1\cup Q_2$ with 
\begin{align*}
	Q_1(x,y,z,w)\leftarrow& R_1(x,y),R_2(y,z),R_3(z,x),R_4(x,w)\text{ and}\\
	Q_2(x,y,z,w)\leftarrow& R_1(x,t_1),R_2(y,t_2),R_3(z,t_3),R_4(w,t_4)
\end{align*}
$Q_1$ is cyclic, so it is not possible to find an answer to it in linear time. However, since $Q_2$ is free-connex and supplies variables that map to the cycle $\{x,y,z\}$ in $Q_1$ via a body-homomorphism, and since $Q_1$ becomes free connex when adding an atom with these variables, the union is solvable in linear preprocessing time and constant delay by saving the results of $Q_2$ as a relation as suggested in Section~\ref{sec:positive}. In this case, there is also a way of solving it with no additional space: use CDY to solve $Q_2$ efficiently; for every answer $(a,b,c,d)$ to $Q_2$, check if $a=d$, $(a,b)\in R_1$, $(b,c)\in R_2$ and $(c,a)\in R_3$; if all these conditions are met, go over all tuples $(a,e)\in R_4$ and output $(a,b,c,e)$.
These checks can be done in constant time per answer in the RAM model by constructing lookup tables for each relation during preprocessing.
This finds all answers to $Q_1\cup Q_2$ with constant time per answer.
However, if an answer belongs to both CQs, we print it twice.
The duplicates can easily be avoided by replacing $R_4$ in $Q_1$ with a modified version, lacking the tuples that cause duplicates.
During preprocessing, compute
$R_4'=R_4\setminus\{(a,e)\in R_4\mid\exists{t}\text{ s.t. }(e,t)\in R_4\}$.
During enumeration, for every qualifying answer $(a,b,c,d)$ to $Q_2$, go over all tuples $(a,e)\in R_4'$ and output $(a,b,c,e)$.
\end{example}

Example~\ref{example:no-space} shows that reducing the space requirements is sometimes possible. We do not know however whether this is the case for all free-connex UCQs. Consider Example~\ref{example:first}. The same approach would mean that for every answer $(a,b,c)$ to $Q_2$ we go over all tuples $(c,d)\in R_3$ and print $(a,c,d)$. The problem here is that we will print many duplicates: every answer $(a,c,d)$ will be produced as many times as there are different $b$ values with $(a,b,c)\in Q_2(I)$.

In conclusion,
UCQs that contain intractable CQs can sometimes be made easy by using virtual atoms. In this article, we proposed instantiating a relation that corresponds to each virtual atom during enumeration. This instantiation can sometimes be avoided, but it requires a different approach, and it is left for future work to find when this can be done.
The discussion above also gives rise to a possible relaxation that can be considered a first step: can we evaluate such UCQs efficiently with only constant writing memory available during enumeration, when we allow for linear partial time and a constant number of duplicates per answer?
Note that in Example~\ref{example:no-space} we do not need this relaxation, so another interesting question is whether this relaxation is useful for our task or can all UCQs that can be answered with this relaxed measure also be answered with linear preprocessing and constant delay without extra space.

\section{Conclusions}\label{sec:conclusions}

In this article we study the enumeration complexity of UCQs with respect to $\DelayClin$. 
We formalized how CQs within a union can make each other easier by supplying variables, and we introduced union extensions. Then, we defined free-connex UCQs, and showed that these are tractable.
In particular, we demonstrated that UCQs containing only intractable CQs may be tractable.

We showed that in case of a union of two difficult CQs or two acyclic body-isomorphic CQs, free-connexity fully captures the tractable cases. Nevertheless, achieving a full classification of UCQs remains an open problem.
Section~\ref{sec:dichotomy} discusses the next steps that are required to fully characterize which UCQs are in $\DelayClin$ and provides examples with unknown complexity. Resolving
Examples~\ref{example:separated}, \ref{example:star}, \ref{example:cyclic-guarded-hard} and \ref{example:newtetra} is a necessary step on the way to a future dichotomy.

In addition to the general settings, we also showed that in three other variations: cardinality dependencies in the schema, UCQs with disequalities and restricted space, we can identify tractable UCQs that contain intractable CQs.
Section~\ref{sec:extensions} offers a preliminary idea on how the results of the article can be applied in these variations, and leaves substantial room for future exploration of this topic.
For example,
Theorem~\ref{thm:dis-positive} in Section \ref{sec:dis} shows the tractability of free-connex UCQs with disequalities, but as Example~\ref{ex:dis-also-pos} demonstrates, additional such queries are tractable. Can we find a characterization that captures all such tractable UCQs with disequalities?

Most of this article ignores space consumption in order to concentrate on time.
An interesting question is how much space is required to achieve the time bounds described here. As Example~\ref{example:no-space} shows, sometimes we can solve UCQs with only constant memory available for writing during the enumeration phase. It is left open whether all free-connex UCQs can be solved under this restriction, and in general whether $\DelayClin=\CDlin$ (see the discussion in Section~\ref{sec:space}).

\begin{acks}
The authors are greatly
thankful to Benny Kimelfeld for his advice.
The work of Nofar Carmeli was supported by the Irwin and Joan Jacobs Fellowship and by the Google PhD Fellowship.
The work of Markus Kr\"oll was supported by the Austrian Science Fund (FWF):
P30930-N35, W1255-N23.
\end{acks}

\bibliographystyle{ACM-Reference-Format}
\bibliography{main}

\appendix

\end{document}